\documentclass[12pt]{amsart}
\usepackage{amssymb,amsmath}

\usepackage{t1enc}
\usepackage[latin1]{inputenc}
\usepackage[german,english]{babel}
\usepackage{amsfonts}
\usepackage[all]{xy}
\usepackage{graphicx}
\usepackage{color}
\usepackage{tikz}
\usetikzlibrary{arrows,decorations.pathmorphing,backgrounds,positioning,fit,petri}

\usepackage{geometry}
\newtheorem{defi}{Definition}[section]
\newtheorem{prop}[defi]{Proposition}

\theoremstyle{remark}
\newtheorem{rem}[defi]{Remark}
\newtheorem{exa}[defi]{Example}

\begin{document}

\title[Iterated LD-Problem in non-associative key establishment]{Iterated LD-Problem in non-associative key establishment}
\author{Arkadius Kalka and Mina Teicher}
\address{Department of Mathematics,
Bar Ilan University,
Ramat Gan 52900,
Israel}
\email{Arkadius.Kalka@rub.de, teicher@math.biu.ac.il}
\urladdr{http://homepage.ruhr-uni-bochum.de/arkadius.kalka/}

\begin{abstract} 
We construct new non-associative key establishment protocols for all left self-distributive (LD), multi-LD-, and mutual LD-systems.
The hardness of these protocols relies on variations of the (simultaneous) iterated LD-problem and its generalizations.
We discuss instantiations of these protocols using generalized shifted conjugacy in braid groups and their quotients,
LD-conjugacy and $f$-symmetric conjugacy in groups. 
We suggest parameter choices for instantiations in braid groups, symmetric groups and several matrix groups.
\end{abstract} 



\maketitle

\section{Introduction}
 In an effort to construct new key establishment protocols (KEPs), which are hopefully harder to break than previously proposed non-commutative schemes, the first author introduced in his PhD thesis \cite{Ka07} (see also \cite{Ka12}) the first non-associative generalization
of the Anshel-Anshel-Goldfeld KEP \cite{AAG99}, which  revolutionized the field of \emph{non-commutative public key cryptography} (PKC) more than ten years ago.
For an introduction to non-commutative public key cryptography we refer to the book by Myasnikov et al. \cite{MSU11}. 
For further motivation and on {\it non-associative PKC} we refer to \cite{Ka12}. It turns out (see \cite{Ka12}) that in the context of AAG-like KEPs for magmas, left self-distributive systems (LD-systems) and their generalizations (like multi-LD-systems) naturally occur. 
A construction that provides KEPs for all LD-, multi-LD- and mutually left distributive systems was presented in \cite{KT13}.
With this method at hand any LD- or multi-LD-system automatically provides a KEP, and we obtain a rich variety of new non-associatiave KEPs
coming from LD-, multi-LD-, and other left distributive systems.  
Here, we propose somehow improved, iterated versions of the KEPs from \cite{KT13}.
\par
\medskip
\emph{Outline.} In section 2 we review LD-, multi-LD-, and  mutually left distributive systems and provide several important examples,
namely LD-conjugacy and $f$-symmetric conjugacy in groups and shifted conjugacy in braid groups.  
Section 3 describes a KEP for all LD-systems, namely an iterated version of Protocol 1 from \cite{KT13}, and we discuss related base problems. 
In section 4 we describe and analyze a KEP which does not only apply for all multi-LD-systems, but also for a big class of partial multi-LD-systems.
This KEP is an iterated version of Protocol 2 from \cite{KT13}. 
In section 5 we discuss instantiations of these general protocols using generalized shifted conjugacy in braid groups. 
In particular, in section 5.3 we discuss a relevant base problem, namely the (subgroup) conjugacy coset problem, which seems to be a relatively new group-theortic problem, apparently first mentioned in \cite{KT13}.
In section 5.4 we propose several concrete instantiations with parameter suggestions in braid and symmetric groups.
Section 6 deals with other instantiations, namely instantiations using $f$-conjugacy (section 6.1)  and instantiations using $f$-symmetric conjugacy (section 6.2) in groups. For both LD-systems we provide concrete instantiations in finite and infinite matrix groups with suggestions for parameter choices.  
\par
\medskip
\emph{Implementation.} All concrete realizations of the KEPs were implemented in MAGMA \cite{BCP97} which also contains an implementation of braid groups following \cite{CK+01}. Implementation details for these non-associative KEPs are provided in \cite{KT13a}. 
 
\section{LD-systems and other distributive systems}

\subsection{Definitions}
\begin{defi} 
A \emph{left self-distributive (LD) system} $(S,*)$ is a set $S$ equipped with a binary operation $*$ on $S$ which satisfies the 
\emph{left self-distributivity law} 
\[ x*(y*z)=(x*y)*(x*z) \quad {\rm for} \,\, {\rm all} \,\, x,y,z\in S. \] 
\end{defi}

\begin{defi} (Section X.3. in \cite{De00})
Let $I$ be an index set. A \emph{multi-LD-system}
 $(S,(*_i)_{i\in I})$ is a set $S$ equipped with a family of binary operations $(*_i)_{i\in I}$ on $S$ such that 
\[ x*_i(y*_jz)=(x*_iy)*_j(x*_iz) \quad {\rm for} \,\, {\rm all} \,\, x,y,z\in S \]
is satisfied for every $i,j$ in $I$. Especially, it holds for $i=j$, i.e., $(S,*_i)$ is an LD-system. If $|I|=2$ then we call $S$ a \emph{bi-LD-system}. 
\end{defi}

More vaguely, we will also use the terms \emph{partial multi-LD-system} and simply \emph{left distributive system} if the laws of a multi-LD-system
are only fulfilled for special subsets of $S$ or if only some of these (left) distributive laws are satisfied. 

\begin{defi} 
A \emph{mutual left distributive system} $(S, *_a, *_b)$ is a set $S$ equipped with two binary operations $*_a, *_b$ on $S$ such that 
\[ x*_a(y*_bz)=(x*_ay)*_b(x*_az) \quad  x*_b(y*_az)=(x*_by)*_a(x*_bz)  \quad {\rm for} \,\, {\rm all} \,\, x,y,z\in S. \]
\end{defi}

A mutual left distributive system $(L, *_a, *_b)$ is only a partial bi-LD-system. 
The left selfdistributivity laws need not hold, i.e., $(L,*_a)$ and $(L, *_b)$ are in general no LD-systems. 

\subsection{Examples}
We list examples of LD-systems, multi-LD-systems and mutual left distributive systems. 
More details can be found in \cite{De00, De06, Ka12, KT13}.

\subsubsection{Trivial example.}  $(S,*)$ with $x*y=f(y)$ is an LD-system for any function $f: S\rightarrow S$. 

\subsubsection{Free LD-systems.} A set $S$ with a binary operation $*$, that satisfies no other relations than those resulting from the left self-distributivity law, is a free LD-system. Free LD-systems are studied extensively in \cite{De00}. 

\subsubsection{Conjugacy.} A classical example of an LD-system is $(G,*)$ where $G$ is a group equipped with the conjugacy operation $x*y=x^{-1}yx$ (or $x*^{\rm rev}y=xyx^{-1}$).
 Note that such an LD-system cannot be free, because conjugacy satisfies additionally the idempotency law $x*x=x$. 

\subsubsection{Laver tables.} Finite groups equipped with the conjugacy operation are not the only finite LD-systems. Indeed, the socalled \emph{Laver tables} provide the classical example for finite LD-systems. 
There exists for each $n\in \mathbb{N}$ an unique LD-system $L_n=(\{ 1, 2, \ldots , 2^n \},*)$ with $k*1=k+1$.
The values for $k*l$ with $l\ne 1$ can be computed by induction using the left self-distributive law.
The Laver tables for $n=1,2,3$ are
\[ \begin{tabular}{c|cc} $L_1$ &1&2 \\ \hline 1&2&2 \\ 2&1&2 \end{tabular} \quad
\begin{tabular}{c|cccc} $L_2$ &1&2&3&4 \\ \hline 1&2&4&2&4 \\ 2&3&4&3&4 \\ 3&4&4&4&4 \\ 4&1&2&3&4 \end{tabular} \quad 
\begin{tabular}{c|cccc cccc} $L_3$ &1&2&3&4&5&6&7&8 \\ \hline 1&2&4&6&8&2&4&6&8 \\ 2&3&4&7&8&3&4&7&8 \\ 3&4&8&4&8&4&8&4&8 \\ 4&5&6&7&8&5&6&7&8 \\
                                                      5&6&8&6&8&6&8&6&8 \\           6&7&8&7&8&7&8&7&8 \\ 7&8&8&8&8&8&8&8&8 \\ 8&1&2&3&4&5&6&7&8
\end{tabular}  \]
Laver tables are also described in \cite{De00}. 

\subsubsection{LD-conjugacy.}
Let $G$ be a group, and $f\in End(G)$. Set 
\[ x*_fy=f(x^{-1}y)x, \] 
then $(G,*_f)$ is an LD-system. 
We call an ordered pair $(u, v)\in G \times G$ \emph{$f$-LD-conjugated} or \emph{LD-conjugated}, or simply \emph{$f$-conjugated}, denoted by $u\stackrel{}{\longrightarrow }_{*_f}v$, if there exists a $c\in G$ such that $v=c*_f u=f(c^{-1}u)c$.
\par
More general,  let $f,g,h \in End(G)$. Then the binary operation
$x*y=f(x^{-1}) \cdot g(y)\cdot h(x)$ yields an LD-structure on $G$ if and only if 
\begin{equation} \label{fConjLDeqs} fh=f, \quad gh=hg=hf, \quad fg=gf=f^2, \quad h^2=h. 
\end{equation}
See Proposition 2.3 in \cite{KT13}.
The simplest solution of the system of equations (\ref{fConjLDeqs}) is $f=g$ and $h={\rm id}$ which leads to the definition of LD-conjugacy given above. 

\subsubsection{Shifted conjugacy.}
Consider the braid group on infinitely many strands
\[ B_{\infty }= \langle \{\sigma _i \}_{i\ge 1} \mid \sigma _i\sigma _j=\sigma _j\sigma _i \,\,{\rm for} \,\, |i-j|\ge 2, \,\, 
\sigma _i\sigma _j\sigma _i=\sigma _j\sigma _i\sigma _j \,\, {\rm for} \,\, |i-j|=1\rangle  \] 
where inside $\sigma _i$ the $(i+1)$-th strand crosses over the $i$-th strand. 
The {\it shift map} $\partial : B_{\infty } \longrightarrow B_{\infty }$ defined by $\sigma _i \mapsto \sigma _{i+1}$ for all $i\ge 1$
is an injective endomorphism. Then
$B_{\infty }$ equipped with the {\rm shifted conjugacy} operations $*$, $\bar{*}$ defined by
\[ x*y=\partial x^{-1}\cdot \sigma _1 \cdot \partial y \cdot x, \quad \quad  x\, \bar{*}\, y=\partial x^{-1}\cdot \sigma _1^{-1} \cdot \partial y  \cdot x  \]
is a bi-LD-system. In particular, $(B_{\infty },*)$ is an LD-system. 
\par
Dehornoy points out, that once the definition of shifted conjugacy is used, braids inevitably appear (see Exercise I.3.20 in \cite{De00}). 
Consider a group $G$, an endomorphism $f \in End(G)$, and a fixed element $a\in G$. Then the binary operation
$x*y=x *_{f,a}y=f(x)^{-1} \cdot a \cdot f(y)\cdot x$
yields an LD-structure on $G$ if and only if $[f^2(x), a]=1$ for all $x\in G$, and $a$ satisfies the relation $a f(a) a=f(a) a f(a)$\footnote{ 
Note that $[\partial ^2(x), \sigma _1]=1$ for all $x\in B_{\infty }$, and 
$\sigma _1 \partial (\sigma _1) \sigma _1=\partial (\sigma _1) \sigma _1 \partial (\sigma _1)$ holds.}.
\par
Hence the subgroup $H=\langle \{ f^n(a) \mid  n\in \mathbb{N} \} \rangle $ of $G$ is a homomorphic image of
the braid group $B_{\infty }$ on infinitely many strands, i.e., up to an isomorphism, it is a quotient of $B_{\infty }$. 
In case of $a=1$ this subgroup $H$ is trivial and the binary operation $*_{f,1}$ becomes $f$-conjugacy. 
\par
There exists a straightforward generalization of Exercise I.3.20 in \cite{De00} for multi-LD-systems:
\par
Let $I$ be an index set. Consider a group $G$, a family of endomorphisms $(f_i)_{i\in I}$ of $G$, and a set of fixed elements $\{a_i\in G \mid i\in I\}$. 
Then $(G,(*_i)_{i\in I})$ with
\[ x*_iy= f_i(x^{-1})\cdot a_i \cdot f_i(y)\cdot x\]
is a multi-LD-system if and only if 
$f_i=f_j=:f$ for all $i\ne j$, 
\begin{equation}  \label{multiLD}  [a_i,f^2(x)]=1 \quad  \forall x\in G, \,\, i\in I, \quad 
{\rm and} \quad a_if(a_i)a_j=f(a_j)a_if(a_i) \quad  \forall i,j\in I. 
\end{equation}
For a proof see, e.g. Proposition 4.6 in \cite{Ka12}. 

\subsubsection{Generalized shifted conjugacy in braid groups.}
In the following we consider generalizations of the shifted conjugacy operations $*$ in $B_{\infty }$. We
set $f=\partial ^p$ for some $p\in \mathbb{N}$, and we choose $a_i\in B_{2p}$ for all $i\in I$ such that 
\begin{equation} \label{genShift} a_i\partial ^p(a_i)a_j=\partial ^p(a_j)a_i\partial ^p(a_i) \quad {\rm for} \,\, {\rm all} \,\,  i,j\in I.  \end{equation}  
Since $a_i\in B_{2p}$, we have $[a_i,\partial ^{2p}(x)]=1$ for all $x\in B_{\infty }$.
Thus the conditions (\ref{multiLD}) are fulfilled, and $x*_iy=x\partial ^p(y)a_i\partial ^p(x^{-1})$ defines a multi-LD-structure on $B_{\infty}$. 
For $|I|=1$, $p=1$ and $a=\sigma _1$, which implies $H=B_{\infty }$, we get Dehornoy's original definition of shifted conjugacy $*$. 
\par
It remains to give some natural solutions $\{a_i\in B_{2p} \mid i\in I \}$ of the equation set (\ref{genShift}). 
Let, for $n\ge 2$,  $\delta _n=\sigma _{n-1}\cdots \sigma _2\sigma _1$. For $p, q \ge 1$, we set
\[ \tau _{p,q}=\delta _{p+1}\partial (\delta _{p+1})\cdots \partial ^{q-1}(\delta _{p+1}). \]
Since $a=\tau _{p,p}^{\pm 1}\in B_{2p}$ fulfills $a\partial ^p(a)a=\partial ^p(a)a\partial ^p(a)$, it provides a lot of (multi)-LD-structures on $B_{\infty }$.  

\begin{prop} \label{abc}
(a) The binary operation $x*_ay=\partial ^p(x^{-1})a\partial ^p(y)x$ with $a=a'\tau _{p,p}a''$ for some $a',a''\in B_p$ yields an LD-structure on $B_{\infty }$ 
if and only if $[a',a'']=1$. 
\par
\noindent (b) Let $I$ be an index set.
The binary operations $x*_iy=\partial ^p(x^{-1})a_i\partial ^p(y)x$ with $a_i=a'_i\tau _{p,p}a''_i$ for some $a'_i,a''_i\in B_p$ ($i\in I$) yields a multi-LD-structure on $B_{\infty }$ 
if and only if $[a_i',a_j']=[a_i',a_j'']=1$ for all $i,j\in I$. (Note that $a_i''$ and $a_j''$ needn't commute for $i\ne j$.) 
\par
\noindent (c) The binary operations $x*_iy=\partial ^p(x^{-1})a_i\partial ^p(y)x$ ($i=1,2$) with $a_1=a_1'\tau _{p,p}a_1''$, $a_2=a_2'\tau _{p,p}^{-1}a_2''$ for some $a_1',a_1'',a_2',a_2''\in B_p$ 
yields a bi-LD-structure on $B_{\infty }$ if and only if $[a_1',a_1'']=[a_2',a_2'']=[a_1',a_2'']=[a_2',a_1'']=[a_1',a_2']=1$.
(Note that $a_1''$ and $a_2''$ needn't commute.)  
\par
\noindent (d)  $(B_{\infty }, *_1, *_2)$ with binary operations $x*_iy=\partial ^p(x^{-1})a_i\partial ^p(y)x$ ($i=1,2$) with $a_1=a_1'\tau _{p,p}a_1''$, $a_2=a_2'\tau _{p,p}^{-1}a_2''$ for some $a_1',a_1'',a_2',a_2''\in B_p$ is a mutual left distributive system 
 if and only if $[a_1',a_2'']=[a_2',a_1'']=[a_1',a_2']=1$.
(Note that $[a_1',a_1'']$, $[a_2',a_2'']$ and $[a_1'',a_2'']$ may be nontrivial.) 
\end{prop}

The proofs are straightforward computations. The reader is recommended to draw some pictures.

\subsubsection{Symmetric conjugacy.} 
For a group $G$, there exists yet another LD-operation. $(G,\circ )$ is an LD-system with 
\[ x\circ y=xy^{-1}x. \]
Note that, contrary to the conjugacy operation $*$, for this \emph{"symmetric conjugacy" operation} $\circ $, 
the corresponding relation $\stackrel{}{\longrightarrow }_{\circ }$, defined by $x\stackrel{}{\longrightarrow }_{\circ }y$ if and only if there exists a $c\in G$ such that $y=c\circ x$, is not an equivalence relation. In particular, $\stackrel{}{\longrightarrow }_{\circ }$ is reflexive and symmetric, 
but not transitive. 

\subsubsection{$f$-symmetric conjugacy.}
One may consider several generalizations of this symmetric conjugacy operation $\circ $, as candidates for natural LD-operations in groups.
Let $G$ be a group, and $f,g,h \in End(G)$. Then the binary operation
$x\circ _{f,g,h}y=f(x) \cdot g(y^{-1})\cdot h(x)$ yields an LD-structure on $G$ if and only if 
\begin{equation} \label{fsymmConjLDeqs} f^2=f, \quad fh=gh=fg, \quad hg=gf=hf, \quad h^2=h. \end{equation}

For a proof see Proposition 2.13 in \cite{KT13}.
\par 
Except for $f^2=f=g=h=h^2$, the simplest solutions of the system of equations (\ref{fsymmConjLDeqs}) are $f^2=f=g$ and $h={\rm id}$, or
$f={\rm id}$ and $g=h=h^2$.  
\par
Let $G$ be a group, and $f\in End(G)$ an endomorphism that is also a projector ($f^2=f$).
Then $(G, \circ _f)$ and $(G, \circ _f^{\rm rev})$, defined by 
\[ x\circ _f y=f(xy^{-1})x \quad {\rm and} \quad x \circ _f^{\rm rev} y=xf(y^{-1}x), \] 
are LD-systems.
\par
We have the following left distributivity results.
\par
\noindent (i) The binary operations $\circ _{f,g,h}$ and $*_{f,g,h}$
are distributive over $\circ $. In particular $*$ is distributive over $\circ $. In short, the following equations hold.
\[   x*_{f,g,h}(y\circ z)=(x*_{f,g,h}y)\circ (x*_{f,g,h}z), \quad x\circ _{f,g,h}(y\circ z)=(x\circ _{f,g,h}y)\circ (\circ _{f,g,h}z) \forall x,y,z\in G.\]
(ii) The operations $\circ _f$ and $*_f$ ($*^{\rm rev}_f$) are distributive over $\circ _g$ if and only if $f=gf=fg$.
\par
From (ii) we conclude that $(G, \circ _f, \circ _g)$ is not a mutual left distributive system for $f\ne g$.

\section{Key establishment for all LD-systems}
\subsection{The protocol}

Recall that a \emph{magma} is a set $M$ equipped with a binary operation, say $\bullet $, which is possibly non-associative. 
For our purposes all interesting LD-systems are non-associative. 
Consider an element $y$ of a magma $(M,\bullet )$ which is an iterated product of other elements in $M$. Such an element can be described by a planar rooted binary
tree $T$ whose $k$ leaves are labelled by these other elements $y_1,\ldots ,y_k\in M$. We use the notation $y=T_{\bullet }(y_1,\ldots ,y_k)$. 
Here the subscript $\bullet $ tells us that the grafting of subtrees of $T$ corresponds to the operation $\bullet $. 
\par 
Consider, for example, the element $y=((b \bullet c) \bullet (a \bullet b)) \bullet b $. The corresponding
labelled planar rooted binary tree $T$ is displayed in the following figure.

\begin{figure}[ht]
  \caption{ The element $y=((b \bullet c) \bullet (a \bullet b)) \bullet b =T_{\bullet }(b,c,a,b,b)$}
\begin{center}
\begin{tikzpicture}
   \node (r1) at (0,0) [circle,inner sep=1pt,draw=black!50,fill=black!20]{}; \node[below] at (r1.south) {$b$};  
   \node (r2) at (2,0) [circle,inner sep=1pt,draw=black!50,fill=black!20]{}; \node[below] at (r2.south) {$c$};  
   \node (r3) at (4,0) [circle,inner sep=1pt,draw=black!50,fill=black!20]{}; \node[below] at (r3.south) {$a$};  
   \node (r4) at (6,0) [circle,inner sep=1pt,draw=black!50,fill=black!20]{}; \node[below] at (r4.south) {$b$};  
   \node (r5) at (8,0) [circle,inner sep=1pt,draw=black!50,fill=black!20]{}; \node[below] at (r5.south) {$b$};  
   \node (i12) at (1,1) [circle,inner sep=1pt,draw=black!50,fill=black!20]{}
          edge[thick] (r1)   edge[thick] (r2);     \node[below] at (i12.south) {$\bullet $};  
   \node (i34) at (5,1) [circle,inner sep=1pt,draw=black!50,fill=black!20]{}
          edge[thick] (r3)   edge[thick] (r4);     \node[below] at (i34.south) {$\bullet $};  
   \node (i14) at (3,3) [circle,inner sep=1pt,draw=black!50,fill=black!20]{}
          edge[thick] (i12)   edge[thick] (i34);     \node[below] at (i14.south) {$\bullet $};  
   \node (i15) at (4,4) [circle,inner sep=1pt,draw=black!50,fill=black!20]{}
          edge[thick] (i14)   edge[thick] (r5);   \node[below] at (i15.south) {$\bullet $}; 
\end{tikzpicture}
\end{center}
\end{figure}

It is easy to prove by induction (over the depth of the involved trees) that any magma homomorphism $\beta :(M,\bullet )\rightarrow (N,\circ )$ satisfies
\[ \beta (T_{\bullet }(y_1,\ldots ,y_k))=T_{\circ }(\beta (y_1),\ldots ,\beta (y_k)) \]
for all $y_1,\ldots ,y_k\in M$. 

\begin{prop} \label{LDendo}
Let $(L,*)$ be an LD-system.
Then, for any element $x\in L$, the left multiplication map $\phi _x: y \mapsto x*y$ defines a magma endomorphism of $L$.
\end{prop}

\begin{proof} 
$\phi _x(y_1*y_2)=x*(y_1*y_2)\stackrel{LD}{=}(x*y_1)*(x*y_2)=\phi _x(y_1) * \phi _x(y_2)$. 
\end{proof}

\begin{prop} \label{itLDendo}
Let $(L,*)$ be an LD-system and $k \in \mathbb{N}$.
Then, for all $x_1, \ldots , x_k \in L$, the iterated left multiplication map 
\[ \phi _{x_1,\ldots , x_k}: y \mapsto x_k*(x_{k-1}* \cdots *(x_2 *(x_1*y)) \cdots )\]
defines a magma endomorphism of $L$.
\end{prop}

\begin{proof} 
Proof by induction over $k$.
\begin{eqnarray*}
 && \phi _{x_1,\ldots , x_k}(y_1*y_2) =x_k* \phi _{x_1,\ldots , x_{k-1}}(y_1*y_2) \stackrel{IH}{=}
 x_k * (\phi _{x_1,\ldots , x_{k-1}}(y_1)*  \phi _{x_1,\ldots , x_{k-1}}(y_2)) \\
 &\stackrel{LD}{=}& (x_k*  \phi _{x_1,\ldots , x_{k-1}}(y_1)) * (x_k*  \phi _{x_1,\ldots , x_{k-1}}(y_2))=
  \phi _{x_1,\ldots , x_k}(y_1)* \phi _{x_1,\ldots , x_k}(y_2). 
\end{eqnarray*}
\end{proof}

We are going to describe a KEP that applies to any LD-system $(L,*)$.
There are two public submagmas $S_A=\langle s_1, \cdots , s_m \rangle _*$, $S_B=\langle t_1, \cdots , t_n \rangle _*$ of $(L,*)$, assigned to Alice and Bob.
Alice and Bob perform the following protocol steps.

\begin{description}
\item[{\bf Protocol 1}] {\sc Key establishment for any LD-system} $(L,*)$.
\item[{\rm 1}] Alice generates her secret key $(a_0,a_1, \ldots , a_{k_A}) \in S_A \times L^{k_A}$, and Bob chooses his secret key $b\in S_B^{k_B}$.
In particular, Alice's and Bob's secret magma morphisms $\alpha $ and $\beta $ are given by
\begin{eqnarray*} 
\alpha (y)&=&a_{k_A}*(a_{k_A-1}* \cdots *(a_2 *(a_1*y)) \cdots ) \quad {\rm and} \\ 
\beta (y)&=& b_{k_B}*(b_{k_B-1}* \cdots *(b_2 *(b_1*y)) \cdots ), 
\end{eqnarray*}
respectively.
\item[{\rm 2}] Alice computes the elements 
$(\alpha (t_i))_{1\le i \le n} \in L^n, p_0= \alpha (a_0) \in L$, and sends them to Bob. 
Bob computes the vector $(\beta (s_j))_{1\le j \le m} \in L^m$, and sends it to Alice. 
\item[{\rm 3}] Alice, knowing $a_0=T_*(r_1, \ldots , r_l)$ with $r_i\in \{s_1,\ldots ,s_m\}$, computes from the received message
\[ T_*( \beta (r_1), \ldots , \beta (r_l))=\beta (T_*(r_1, \ldots , r_l))=\beta (a_0). \]
And Bob, knowing for all $1\le j \le k_B$, $b_j=T^{(j)}_*(u_{j,1}, \ldots , u_{j,l_j})$ with $u_{j,i}\in \{t_1,\ldots ,t_n\} \forall i\le l_j$ for some
$l_j \in \mathbb{N}$, computes from his received message for all $1\le j \le k_B$
\[ T^{(j)}_*(\alpha (u_{j,1}), \ldots , \alpha (u_{j,l_j}))=\alpha (T^{(j)}_*(u_{j,1}, \ldots , u_{j,l_j})=\alpha (b_j). \]
\item[{\rm 4}] Alice computes $K_A=\alpha (\beta (a_0))$.
Bob gets the shared key by 
\[ K_B:=\alpha (b_{k_B})*( \alpha (b_{k_B-1})*( \cdots (\alpha (b_2)*(\alpha (b_1)*p_0)) \cdots ))\stackrel{(LD)}{=}K_A. \]
\end{description}

This protocol is an iterated version of Protocol 1 in \cite{KT13} and an
asymmetric modification of the Anshel-Anshel-Goldfeld protocols for magmas introduced in \cite{Ka07, Ka12}. 

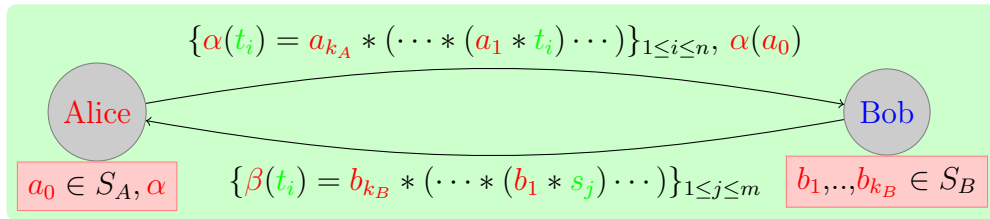
\begin{figure}[ht]
  \caption{\sc Protocol 1: Key establishment for any LD-system - iterated version.}
\begin{center} 
 \begin{tikzpicture}
 \node[red]  (Alice) at ( 0,0) [circle,draw=black!50,fill=black!20]{Alice};
 \node[blue]  (Bob) at ( 10.5,0) [circle,draw=black!50,fill=black!20]{Bob}
   edge [<-, bend right=10] node[auto,swap] (pA) {$\{ {\color{red}\alpha } ({\color{green}t_i})={\color{red}a_{k_A} }* (\cdots *({\color{red}a_1 } *{\color{green}t_i})\cdots ) \}_{1\le i \le n}, \, {\color{red}\alpha }({\color{red}a_0}) $} (Alice)   
   edge [->, bend left=10] node[auto] (pB) { $\{ {\color{red}\beta } ({\color{green}t_i})={\color{red}b_{k_B} }*(\cdots *({\color{red}b_1 } *{\color{green}s_j}) \cdots )\}_{1\le j \le m} $} (Alice) ;
 \node (skA) [below] at (Alice.south) [rectangle,draw=red!50,fill=red!20]{${\color{red}a_0}\in S_A, {\color{red}\alpha }$};
 \node (skB) [below] at (Bob.south) [rectangle,draw=red!50,fill=red!20]{${\color{red}b_1}$,..,${\color{red}b_{k_B}}\in S_B$};
 \begin{pgfonlayer}{background}
  \node [fill=green!20, rounded corners, fit=(Alice) (Bob) (skA) (skB) (pA) (pB)] {};
 \end{pgfonlayer} 
 \end{tikzpicture}
\end{center}
\end{figure}

\subsection{Base problems} 
In order to break Protocol 1 an attacker has to find the shared key $K=K_A=K_B$.
A successful attack on Bob's secret key $b$ requires the solution of
\begin{list}{}{\setlength{\itemsep}{0cm} \setlength{\parsep}{0cm} }
\item[{\bf $m$-simItLDP} ($m$-simultaneous iterated LD-Problem):]
\item[{\sc Input:}]  Element pairs $(s_1,s'_1),\ldots ,(s_m,s'_m)\in L^2$ with 
$s'_i=\phi _{b_1,\ldots , b_{k_B}}(s_i)$ $\forall 1\le i\le m$ for some (unknown) $b_1, \ldots , b_{k_B} \in L$ and some $k_B\in \mathbb{N}$.  
\item[{\sc Objective:}] Find $k'_B\in \mathbb{N}$, $b'_1, \ldots , b'_{k_B} \in L$  such that
\[ s'_i=b'_{k'_B}*(b'_{k'_B-1}* (\cdots *(b'_2 *(b'_1*s_i)) \cdots )) \quad \forall \, i=1,\ldots ,m.  \]
\end{list}

Note that in our context, $b$ comes from a restricted domain, namely $S_B^{k_B} \subseteq L^{k_B}$. 
This might affect distributions when one considers possible attacks.
Nevertheless, we use the notion of (simultaneous) iterated LD-Problem for inputs generated by potentially arbitrary $b\in L^{k_B}$. 
Similar remarks affect base problems further in the text. 
\par 
Even if an attacker finds Bob's original key $b\in S_B^{k_B}$ or a pseudo-key $b' \in S_B^{k'_B} \subseteq \bigcup _{i=1}^{\infty } S_B^i$ 
(solution to the $m$-simItLDP above), then she still faces the following problem for all $i=1, \ldots , k'_B$.
\begin{list}{}{\setlength{\itemsep}{0cm} \setlength{\parsep}{0cm} }
\item[{\bf $*$-MSP} ($*$-submagma Membership Search Problem):]
\item[{\sc Input:}] $t_1,\ldots ,t_n\in (L,*)$, $b'_i\in \langle t_1, \ldots , t_n \rangle _*$.
\item[{\sc Objective:}] Find an expression of $b'_i$ as a tree-word in the submagma $\langle t_1,\ldots ,t_n \rangle _*$ (notation $b'_i=T_*(u_1,\ldots ,u_l)$ for $u_1, \ldots , u_l \in \{ t_j \} _{j\le n}$).   
\end{list}

\begin{prop}  \label{BaseProbs1KEP1}
Let $(L,*)$ be an LD-system. We define the generalized $m$-simItLDP for $S_B\subseteq L$ as an $m$-simultaneous iterated LD-Problem 
with the objective to find $k'_B \in \mathbb{N}$ and  $b'\in S_B^{k'_B}$ with $S_B=\langle t_1, \ldots , t_n \rangle _*$ such that 
$\phi _{b'}:=\phi _{b'_1,\ldots , b'_{k'_B}}(s_i)=s'_i$ for all $i\le m$. 
\par An oracle that solves the generalized $m$-simItLDP and $*$-MSP for $S_B$ is sufficient to break key establishment Protocol 1.
\end{prop}

\begin{proof} 
As outlined above, we perform an attack on Bob's private key. The generalized $m$-simItLDP oracle provides a pseudo-key vector
$b'\in S_B^{k'_B}$ with  $\phi _{b'}(s_i)=s'_i=\phi _b(s_i)$ for all $i=1,\ldots ,m$. 
Observe that this implies for any element $e_A\in S_A$ that $\phi _{b'} (e_A)= \phi _b (e_A)$.
In particular, we have $\phi _{b'}(a_0)=\phi _b(a_0)$.
For $i=1, \ldots , k'_B$, we feed each pseudo-key component $b'_i$ into a $*$-MSP oracle for $S_B$ which returns a treeword  
$T^{(i)'}_*(u_{i,1},\ldots ,u_{i,l_i})=b'$ (for some $l\in \mathbb{N}$ and $u_{i,j} \in \{ t_k \} _{k\le n}$). Now compute, for each $1\le i \le k'_B$, 
\[
T^{(i)'}_*(\alpha (u_{i,1}),\ldots , \alpha (u_{i,l_i})) \stackrel{LD}{=} \alpha (T^{(i)'}_*(u_{i,1},\ldots ,u_{i,l_i}))= 
\alpha (b'_i).
\]
This enables us to compute
\begin{eqnarray*}
K'_B &=& \alpha (b'_{k'_B})*(\alpha (b'_{k'_B-1})* (\cdots *(\alpha (b'_2) *(\alpha (b'_1)*p_0)) \cdots )) \\
&\stackrel{\alpha \,\, {\rm hom}}{=}& \alpha (b'_{k'_B}*(b'_{k'_B-1}* (\cdots *(b'_2 *(b'_1*a_0)) \cdots ))=\alpha (\phi _{b'}(a_0))=
\alpha (\beta (a_0))=K_A.
\end{eqnarray*}
\end{proof}

Note that here the situation is asymmetric - an attack on Alice's secret key requires the solution of the following problem.
\begin{list}{}{\setlength{\itemsep}{0cm} \setlength{\parsep}{0cm} }
\item[{\bf $n$-modsimItLDP} (Modified $n$-simultaneous iterated LD-Problem):]
\item[{\sc Input:}]  An element $p_0\in L$ and pairs $(t_1,t'_1),\ldots ,(t_n,t'_n)\in L^2$ with 
$t'_i=\phi _a(t_i)$ $\forall 1\le i\le n$ for some (unknown) $k_A \in \mathbb{N}$ and $a \in L^{k_A}$.  
\item[{\sc Objective:}] Find $k'_A\in \mathbb{N}$ and elements $a'_0, a'\in L^{k'_A}$ such that $p_0=\phi _{a'}(a'_0)$ and 
$\phi _{a'}(t_i)=t'_i$ for all $i=1,\ldots ,n$.
\end{list}

Also here, even if an attacker finds Alice's original key $(a_0,a)$ or a pseudo-key $(a'_0, a') \in S_A \times L$, 
then she still faces a $*$-submagma Membership Search Problem.

\begin{prop} \label{BaseProbs2KEP1}
Let $(L,*)$ be an LD-system. We define the generalized $n$-modsimItLDP for $S_A\subseteq L$ as a modified $n$-simultaneous iterated LD-Problem 
with the objective to find $k'_A\in \mathbb{N}$, $a'\in L^{k'_A}$ and $a'_0$ in $S_A=\langle s_1, \ldots , s_m \rangle _*$ 
such that $\phi _{a'}(a'_0)=p_0$ and $\phi _{a'}(t_i)=t'_i$ for all $i\le n$. 
\par An oracle that solves the generalized $n$-modsimLDP and $*$-MSP for $S_A$ is sufficient to break key establishment Protocol 1.
\end{prop}

\begin{proof} 
As outlined above, we perform an attack on Alice's private key. The generalized $n$-modsimItLDP oracle provides a pseudo-key $(a'_0, a')'\in S_A \times L^{k'_A}$ (for some $k'_A\in \mathbb{N}$) such that $\phi _{a'}(a'_0)=p_0$ and $\phi _{a'}(t_i)=a'_i=a*t_i$ for all $i=1,\ldots ,n$. Observe that this implies for any element $e_B\in S_B$ that $\phi _{a'}(e_B)=\phi _a(e_B)=:\alpha (e_B)$.
In particular, we have $\phi _{a'}(b_j)=\alpha (b_j)$ for all $1\le j \le k_B$.
We feed the first component $a'_0 \in S_A$ of this pseudo-key into a $*$-MSP oracle for $S_A$ which returns a treeword  $T'_*(r_1,\ldots ,r_l)=a'_0$ (for some $l\in \mathbb{N}$ and $r_i \in \{ s_j \} _{j\le m}$). Now, we compute 
\begin{eqnarray*}
K'_A &=& \phi _{a'}(T'_*(\beta (r_1),\ldots ,\beta (r_l))) \stackrel{LD}{=} \phi _{a'}(\beta (T'_*(r_1,\ldots ,r_l)))= \phi _{a'}(\beta (a'_0)) \\
&=& \phi _{a'}(b_{k_B}* (\cdots *(b_2 *(b_1*a'_0))\cdots ))  \\
&=& \phi _{a'}(b_{k_B})* (\cdots *(\phi _{a'}(b_2) *(\phi _{a'}(b_1)*\phi _{a'}(a'_0)))\cdots ) \\
&=&   \phi _{a}(b_{k_B})* (\cdots *(\phi _{a}(b_2) *(\phi _{a}(b_1)*p_0))\cdots )=K_B.  
\end{eqnarray*}
\end{proof}

Both appproaches described above require the solution of a $*$-submagma Membership Search Problem. Note that we assumed that the generalized $m$-simItLDP 
(resp. $n$-modsimItLDP) oracle already provides a pseudo-key in the submagma $S_B$ (resp. $S_A$) which we feed to the $*$-MSP oracle.
But to check whether an element lies in some submagma, i.e. the $*$-submagma Membership Decision Problem, is already undecidable in general. 
\par Fortunately, for the attacker, there are approaches which do not resort to solving the $*$-MSP.

\medskip

Recall that we defined the generalized $m$-simItLDP for $S_B\subseteq L$ as an $m$-simultaneous iterated LD-Problem 
with the objective to find $k'_B\in \mathbb{N}$ and $b'$ in $S_B^{k'_B}$ such that $\phi _{b'}(s_i)=s'_i$ for all $i\le m$.
\begin{prop}  \label{BaseProbs3KEP1} 
A generalized simItLDP oracle is sufficient to break key establishment Protocol 1.
More precisely, an oracle that solves the generalized $m$-simItLDP for $S_B$ {\it and} the $n$-simItLDP is sufficient to break Protocol 1.
\end{prop}

\begin{proof} 
Here we perform attacks on Alice's and Bob's private keys - though we need only a pseudo-key for the second component $a'$ of Alice's key. 
The $n$-simItLDP oracle provides $a'\in L^{k'_A}$ s.t. $\phi _{a'}(t_j)=t'_j=\alpha (t_j)$ for all $j\le n$. 
And the generalized $m$-simItLDP oracle returns the pseudo-key $b'\in S_B^{k'_B}$
s.t. $\phi _{b'}(s_i)=s'_i=\beta (s_i)$ for all $i\le m$. Since $b'\in S_B^{k'_B}$, we conclude that $\phi _{a'}(b'_i)=\alpha (b'_i)$ for all $1\le i \le k'_B$.
Also, $a_0\in S_A$ implies, of course, $\phi _{b'}(a_0)=\beta (a_0)$.
Now, we may compute 
\begin{eqnarray*}
K'_B&=& \phi _{a'}(b'_{k'_B})*( \cdots *(\phi _{a'}(b'_2) *(\phi _{a'}(b'_1)*p_0)) \cdots ) \\
&=& \alpha (b'_{k'_B})*( \cdots *(\alpha (b'_2) *(\alpha (b'_1)* \alpha (a_0))) \cdots ) \\
&\stackrel{\alpha \,\, {\rm hom}}{=}& \alpha (b'_{k'_B}*(\cdots *(b'_2 *(b'_1*a_0)) \cdots )=\alpha (\phi _{b'}(a_0))=
\alpha (\beta (a_0))=K_A.
\end{eqnarray*}
\end{proof}

Recall that we defined the generalized $n$-modsimItLDP for $S_A\subseteq L$ as an $n$-simultaneous iterated LD-Problem 
with the objective to find a $a'_0$ in $S_A=\langle s_1, \ldots , s_m \rangle _*$ (and $a'\in L^{k'_A}$) such that $\phi _{a'}(a'_0)=p_0$ etc.

\begin{prop}  \label{BaseProbs4KEP1} 
An oracle that solves the generalized $n$-modsimItLDP for $S_A$ {\it and} the $m$-simItLDP is sufficient to break Protocol 1.
\end{prop} 

\begin{proof}
Also here we perform attacks on Alice's and Bob's private keys.  
The $m$-simItLDP oracle provides $k'_B\in \mathbb{N}$ and $b'\in L^{k'_B}$ s.t. $\phi _{b'}(s_j)=s'_j=\beta (s_j)$ for all $j\le m$. 
And the generalized $n$-modsimItLDP oracle returns the pseudo-key $(a'_0, a')\in S_A \times L^{k'_A}$
s.t. $\phi _{a'}(t_i)=t'_i=\alpha (t_i)$ for all $i\le n$ {\it and} $\phi _{a'}(a'_0)=p_0$. Since $a'_0\in S_A$, we conclude that $\phi _{b'}(a'_0)=\beta (a'_0)$.
Also, $b\in S_B^{k'_B}$ implies, of course, $\phi _{a'}(b_j)=\alpha (b_j)$ for all $1\le j \le k'_A$.
Now, we compute 
\begin{eqnarray*}
K'_A &=& \phi _{a'}( \phi _{b'}(a'_0)) =\phi _{a'}(b'_{k'_B}* (\cdots *(b'_2 *(b'_1*a'_0)) \cdots ))  \\
&=& \phi _{a'}(b_{k_B})* (\cdots *(\phi _{a'}(b_2) *(\phi _{a'}(b_1)*\phi _{a'}(a'_0)))\cdots ) \\
&=&   \alpha (b_{k_B})* (\cdots *(\alpha (b_2) *(\alpha (b_1)*p_0))\cdots )=K_B. 
\end{eqnarray*}
\end{proof}

\begin{rem}
Note that in the non-associative setting the case $m=n=1$ is of particular interest, i.e. we may \emph{abandon simultaneity} in our base problems
since the submagmas generated by one element are still complicated objects. 
\end{rem}

\section{Key establishment for mutual left distributive systems}
\subsection{The protocol}  \label{KEPpartial}
Here we describe a generalization of Protocol 1 that works for all mutual left distributive systems, in particular all multi-LD-systems.
Consider a set $L$ equipped with a pool of binary operations $O_A \cup O_B$ ($O_A$ and $O_B$ non-empty) s.t.
the operations in $O_A$ are distributive over those in $O_B$ and vice versa, i.e. the following holds
for all $x,y,z\in L$, $*_{\alpha } \in O_A$ and $*_{\beta }\in O_B$.
\begin{eqnarray}
  x*_{\alpha }(y*_{\beta }z)&=&(x*_{\alpha }y)*_{\beta }(x*_{\alpha }z), \,\, {\rm and}    \label{abLD} \\
  x*_{\beta }(y*_{\alpha }z)&=&(x*_{\beta }y)*_{\alpha }(x*_{\beta }z).                    \label{baLD}
\end{eqnarray}
Then $(L, *_{\alpha }, *_{\beta })$ is a mutual left distributive system for all $(*_{\alpha },*_{\beta })\in O_A \times O_B$.
Note that, if $O_A \cap O_B \ne \emptyset$, then $(L, O_A \cap O_B)$ is a multi-LD-system. 
\par Let $s_1, \ldots , s_m, t_1, \ldots , t_n\in L$ be some public elements. We denote
$S_A=\langle s_1, \cdots , s_m \rangle _{O_A}$ and $S_B=\langle t_1, \cdots , t_n \rangle _{O_B}$, two submagmas of $(L,O_A\cup O_B)$.
For example, an element $y$ of $S_A$ can be described by a planar rooted binary
tree $T$ whose $k$ leaves are labelled by these other elements $r_1,\ldots ,r_k$ with $r_i \in \{s_i\}_{i\le m}$.
Here the tree contains further information, namely to each internal vertex we assign a binary operation $*_i \in O_A$.
We use the notation $y=T_{O_A}(r_1,\ldots ,r_k)$. 
The subscript $O_A $ tells us that the grafting of subtrees of $T$ corresponds to the operation $*_i\in O_A$.
Consider, for example, the element $y=((s_3*_{\alpha _2}s_3)*_{\alpha _4}s_1)*_{\alpha _1}(s_2*_{\alpha _2}s_1)$. The corresponding
labelled planar rooted binary tree $T$ is displayed in the following figure.

\begin{figure}[ht]
  \caption{ The element $y=((s_3*_{\alpha _2}s_3)*_{\alpha _4}s_1)*_{\alpha _1}(s_2*_{\alpha _2}s_1) \in S_A$}
\begin{center}
\begin{tikzpicture}
   \node (r1) at (0,0) [circle,inner sep=1pt,draw=black!50,fill=black!20]{}; \node[below] at (r1.south) {$s_3$};  
   \node (r2) at (2,0) [circle,inner sep=1pt,draw=black!50,fill=black!20]{}; \node[below] at (r2.south) {$s_3$};  
   \node (r3) at (4,0) [circle,inner sep=1pt,draw=black!50,fill=black!20]{}; \node[below] at (r3.south) {$s_1$};  
   \node (r4) at (6,0) [circle,inner sep=1pt,draw=black!50,fill=black!20]{}; \node[below] at (r4.south) {$s_2$};  
   \node (r5) at (8,0) [circle,inner sep=1pt,draw=black!50,fill=black!20]{}; \node[below] at (r5.south) {$s_1$};  
   \node (i12) at (1,1) [circle,inner sep=1pt,draw=black!50,fill=black!20]{}
          edge[thick] (r1)   edge[thick] (r2);     \node[below] at (i12.south) {$*_{\alpha _2}$};  
   \node (i13) at (2,2) [circle,inner sep=1pt,draw=black!50,fill=black!20]{}
          edge[thick] (i12)   edge[thick] (r3);     \node[below] at (i13.south) {$*_{\alpha _4}$};  
   \node (i45) at (7,1) [circle,inner sep=1pt,draw=black!50,fill=black!20]{}
          edge[thick] (r4)   edge[thick] (r5);     \node[below] at (i45.south) {$*_{\alpha _2}$};  
   \node (i15) at (4,4) [circle,inner sep=1pt,draw=black!50,fill=black!20]{}
          edge[thick] (i13)   edge[thick] (i45);   \node[below] at (i15.south) {$*_{\alpha _1}$}; 
\end{tikzpicture}
\end{center}
\end{figure}
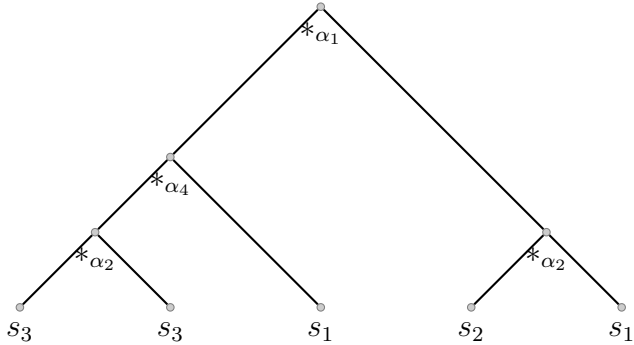

Let $*_{\alpha }\in O_A$ and $*_{\beta }\in O_B$. By induction over the tree depth, it is easy to show that, for all elements $e, e_1, \ldots , e_l \in (L, O_A \cup O_B)$ and all planar rooted binary trees $T$ with $l$ leaves, the following equations hold.
\begin{eqnarray}
e*_{\alpha }T_{O_B}(e_1, \ldots , e_l)&=&T_{O_B}(e*_{\alpha }e_1, \ldots , e*_{\alpha }e_l),   \\
e*_{\beta }T_{O_A}(e_1, \ldots , e_l)&=&T_{O_A}(e*_{\beta }e_1, \ldots , e*_{\beta }e_l).
\end{eqnarray}
Analogeously to Proposition \ref{itLDendo} one may show the following.
\begin{prop} \label{itLDendoMutual} 
Consider $(L,O_A \cup O_B)$ such that $(L, *_A, *_B)$ is a mutual left distributive system for all $(*_A, *_B)\in O_A \times O_B$, 
and let $k \in \mathbb{N}$.
Then, for all $x=(x_1, \ldots , x_k) \in L^k$, $o_A=(*_{A_1}, \ldots ,*_{A_k}) \in O_A^k$, and $o_B=(*_{B_1}, \ldots ,*_{B_k}) \in O_B^k$, 
the iterated left multiplication maps
\begin{eqnarray*} 
\phi _{(x,o_A)}: && y \mapsto x_k*_{A_k}(x_{k-1}*_{A_{k-1}} \cdots *_{A_3}(x_2 *_{A_2}(x_1*_{A_1}y)) \cdots ) \,\, {\rm and}  \\
\phi _{(x,o_B)}: && y \mapsto x_k*_{B_k}(x_{k-1}*_{B_{k-1}} \cdots *_{B_3}(x_2 *_{B_2}(x_1*_{B_1}y)) \cdots )
\end{eqnarray*}
define a magma endomorphisms of $(L, O_B)$ and $(L, O_A)$, respectively.
\end{prop}

In particular, the following equations hold for all $k,l \in mathbb{N}$, $a, b \in L^k$, $o_A\in O_A^k$, $o_B\in O_B^k$,
$e, e_1, \ldots , e_l \in L$ and all planar rooted binary trees $T$ with $l$ leaves.
\begin{eqnarray}
\phi _{(a,o_A)} (T_{O_B}(e_1, \ldots , e_l))&=&T_{O_B}(\phi _{(a,o_A)}(e_1), \ldots , \phi _{(a,o_A)}(e_l)),   \\
\phi _{(b,o_B)} (T_{O_A}(e_1, \ldots , e_l))&=&T_{O_A}(\phi _{(b,o_B)}(e_1), \ldots , \phi _{(b,o_B)}(e_l))
\end{eqnarray}

Now, we are going to describe a KEP that applies to any system $(L,O_A\cup O_B)$ as described above.
We have two subsets of public elements $\{ s_1, \cdots , s_m \}$ and $\{t_1, \cdots , t_n \}$ of $L$.
Also, recall that $S_A=\langle s_1, \cdots , s_m \rangle _{O_A}$ and $S_B=\langle t_1, \cdots , t_n \rangle _{O_B}$.
Alice and Bob perform the following protocol steps.

\medskip

\begin{description}
\item[{\bf Protocol 2}] {\sc Key establishment for the partial multi-LD-system} \par
$(L,O_A\cup O_B)$.
\item[{\rm 1}] Alice generates her secret key $(a_0,a, o_A) \in S_A \times L^{k_A} \times O_A^{k_A}$, and Bob chooses his secret key 
   $(b, o_B)\in S_B^{k_B} \times O_B^{k_B}$. Denote $o_A=(*_{A_1}, \ldots , *_{A_{k_A}})$ and $o_B=(*_{B_1}, \ldots , *_{B_{k_B}})$, 
 then Alice's and Bob's secret magma morphisms $\alpha $ and $\beta $ are given by
\begin{eqnarray*} 
\alpha (y)&=&a_{k_A}*_{A_{k_A}}(a_{k_A-1}*_{A_{k_A-1}} \cdots *_{A_3}(a_2 *_{A_2}(a_1*_{A_1}y)) \cdots ) \quad {\rm and} \\ 
\beta (y)&=&b_{k_B}*_{B_{k_B}}(b_{k_B-1}*_{B_{k_B-1}} \cdots *_{B_3}(b_2 *_{B_2}(b_1*_{B_1}y)) \cdots ),
\end{eqnarray*}
respectively.
\item[{\rm 2}] $(\alpha (t_i))_{1\le i \le n} \in L^n, p_0= \alpha (a_0) \in L$, and sends them to Bob. 
Bob computes the vector $(\beta (s_j))_{1\le j \le m} \in L^m$, and sends it to Alice. 
\item[{\rm 3}] Alice, knowing $a_0=T_{O_A}(r_1, \ldots , r_l)$ with $r_i\in \{s_1,\ldots ,s_m\}$, computes from the received message
\[ T_{O_A}( \beta (r_1), \ldots , \beta (r_l))=\beta (T_{O_A}(r_1, \ldots , r_l))=\beta (a_0). \]
And Bob, knowing for all $1\le j \le k_B$, $b_j=T^{(j)}_{O_B}(u_{j,1}, \ldots , u_{j,l_j})$ with $u_{j,i}\in \{t_1,\ldots ,t_n\} \forall i\le l_j$ for some
$l_j \in \mathbb{N}$, computes from his received message for all $1\le j \le k_B$
\[ T^{(j)}_{O_B}(\alpha (u_{j,1}), \ldots , \alpha (u_{j,l_j}))=\alpha (T^{(j)}_{O_B}(u_{j,1}, \ldots , u_{j,l_j})=\alpha (b_j). \]
\item[{\rm 4}] Alice computes $K_A=\alpha (\beta (a_0))$.
Bob gets the shared key by 
\[ K_B:=\alpha (b_{k_B})*( \alpha (b_{k_B-1})*( \cdots (\alpha (b_2)*(\alpha (b_1)*p_0)) \cdots ))\stackrel{\alpha \, {\rm homo}}{=}K_A. \]
\end{description}

\begin{figure}[ht]
  \caption{{\sc KEP for the partial multi-LD-system} $(L,O_A\cup O_B)$.}
\begin{center} 
 \begin{tikzpicture}
 \node[red]  (Alice) at ( 0,0) [circle,draw=black!50,fill=black!20]{Alice};
 \node[blue]  (Bob) at ( 10,0) [circle,draw=black!50,fill=black!20]{Bob}
   edge [<-, bend right=10] node[auto,swap] (pA) {$\{ \phi _{{\color{red}(a,o_A)}}({\color{green}t_i}) \}_{1\le i \le n}, \,\,\, 
   \phi _{{\color{red}(a,o_A)}}({\color{red}a_0})$} (Alice)   
   edge [->, bend left=10] node[auto] (pB) { $\{ \phi _{{\color{red}(b,o_B)}}({\color{green}s_j}) \}_{1\le j \le m} $} (Alice) ;
 \node (skA) [below] at (Alice.south) [rectangle,draw=red!50,fill=red!20]{${\color{red}a_0}\in S_A, {\color{red}a} \in L^{{\color{red}k_A}}, {\color{red}o_A} \in O_A^{{\color{red}k_A}}$};
 \node (skB) [below] at (Bob.south) [rectangle,draw=red!50,fill=red!20]{${\color{red}b}\in S_B^{{\color{red}k_B}}, {\color{red}o_B} \in O_B^{{\color{red}k_B}}$};
 \begin{pgfonlayer}{background}
  \node [fill=green!20, rounded corners, fit=(Alice) (Bob) (skA) (skB) (pA) (pB)] {};
 \end{pgfonlayer} 
 \end{tikzpicture}
\end{center}
\end{figure}
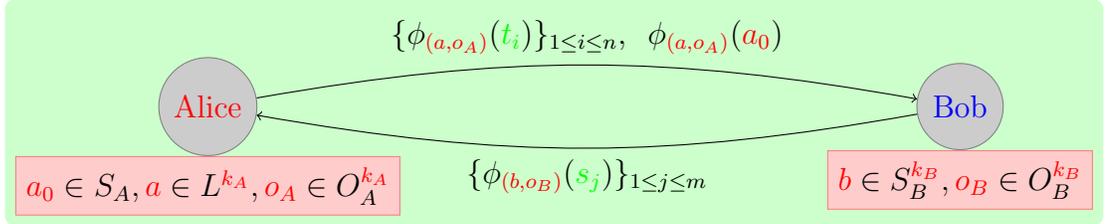

Here the operation vectors $o_A \in O_A^{k_A}$ and $o_B \in O_B^{k_B}$ are part of Alice's and Bob's private keys.
As in Protocol 1, explicit expressions of $a_0\in S_A$ and all $b_i \in S_B$ as treewords $T, T^{(i)}$ (for all $1\le i \le k_B$) 
are also parts of the private keys - though we did not mention it explicitly in step 1 of the protocols. But here $T_{O_A}$ and $T'_{O_B}$ also contain all the information about the grafting operations (in $O_A$ or $O_B$, respectively) at the internal vertices of $T$, $T^{(1)}, \ldots , T^{(k_B)}$.

\subsection{Base problems}

In order to break Protocol 2 an attacker has to find the shared key $K=K_A=K_B$.
A successful attack on Bob's secret key $(b, o_B)\in S_B^{k_B} \times O_B^{k_B}$ requires (first) the solution of the following problem.
\begin{list}{}{\setlength{\itemsep}{0cm} \setlength{\parsep}{0cm} }
\item[{\bf HomSP} (Homomorphism Search Problem for $S_A$):]
\item[{\sc Input:}]  Element pairs $(s_1,s'_1),\ldots ,(s_m,s'_m)\in L^2$ with $s'_i=\phi _{(b,o_B)}(s_i)$ $\forall 1\le i\le m$ for some (unknown) 
$k_B\in \mathbb{N}$, $b \in L^{k_B}$, $o_B\in O_B^{k_B}$.  
\item[{\sc Objective:}] Find $k'_B\in \mathbb{N}$, $b'\in L^{k'_B}$ and $o'_B\in O_B^{k'_B}$, defining a magma homomorphism
$\phi _{(b', o'_B)}: S_A \longrightarrow (L,O_A)$, such that $\phi _{(b', o'_B)}(s_i)=s'_i$ for all $i=1,\ldots ,m$.
\end{list}

Recall that we work in the left distributive system $(L, O_A \cup O_B)$.
Denote 
\begin{eqnarray*}
 Hom(S_A)&=&Hom(S_A,(L,O_A))  \\
 &=&\{ \phi : L \longrightarrow L \mid \phi (y_1*_Ay_2)=\phi (y_1) *_A \phi (y_2) \,\, \forall y_1,y_2 \in S_A \, \forall *_A\in O_A  \}. 
\end{eqnarray*}
We define the {\it generalized HomSP} for $(S_A, S_B)$ as a Homomorphism Search Problem for $S_A$ 
with the objective to find a magma homomorphism $\phi _{(b', o'_B)}\in Hom(S_A)$ with $o'_B\in O_B^{k'_B}$ (for some $k'_B\in \mathbb{N}$) {\it and} 
$b'\in S_B^{k'_B}$ (with $S_B=\langle t_1, \ldots , t_n \rangle _{O_B}$).  

\medskip

Even if an attacker finds a pseudo-key homomorphism $\phi_{(b', o'_B} \in Hom(S_A)$, then she still faces the following problem.
\begin{list}{}{\setlength{\itemsep}{0cm} \setlength{\parsep}{0cm} }
\item[{\bf $O_B$-MSP} ($O_B$-submagma Membership Search Problem for $S_B$):]
\item[{\sc Input:}] $t_1,\ldots ,t_n\in L$, $b\in S_B=\langle t_1, \ldots , t_n \rangle _{O_B}$.
\item[{\sc Objective:}] Find an expression of $b$ as a tree-word (with internal vertices labelled by operations in $O_B$) in the submagma $S_B$ (notation \\ $b=T_{O_B}(u_1,\ldots ,u_k)$ for $u_i \in \{ t_j \} _{j\le n}$).  
\end{list}

\begin{prop}  \label{BaseProbs1KEP2}
An oracle that solves the generalized HomSP for $(S_A, S_B)$ and $O_B$-MSP for $S_B$ is sufficient to break key establishment Protocol 2.
\end{prop}

\begin{proof}
As outlined above, we perform an attack on Bob's private key. The generalized HomSP-oracle for $(S_A, S_B)$ provides a $k'_B\in \mathbb{N}$ and a pseudo-key homomorphism $\phi _{(b', o'_B)}\in Hom(S_A)$ with  $b'\in S_B^{k'_B}$, $o'_B\in O_B^{k'_B}$ such that  
$\phi _{(b', o'_B)}(s_i)=s'_i=\beta (s_i)$ for all $i=1,\ldots ,m$. 
Observe that this implies for any element $e_A\in S_A$ that $\phi _{(b', o'_B)}(e_A)=\beta (e_A)$.
In particular, we have $\phi _{(b', o'_B)}(a_0)=\beta (a_0)$.
Since $b'\in S_B^{k'_B}$, we may feed, for each $1\le i \le k'_B$, $b'_i$ into a $O_B$-MSP oracle for $S_B$ which returns a tree-word  
$T^{(i)'}_{O_B}(u_{i,1},\ldots ,u_{i,l_i})=b'_i$ (for some $l_i\in \mathbb{N}$ and $u_{i,j} \in \{ t_k \} _{k\le n}$). 
Now, we compute for each $1\le i \le k'_B$,
\[ T^{(i)'}_{O_B}(\alpha (u_1),\ldots ,\alpha (u_l)) \stackrel{\alpha \, {\rm homo}}{=}  \alpha (T^{(i)'}_{O_B}(u_1,\ldots ,u_l))= \alpha (b'_i). \]
Let $o'_B=(*_{B_1},\ldots ,*_{B_{k'_B}})$.
This enables us to compute
\begin{eqnarray*}
K'_B &=& \alpha (b'_{k'_B})*_{B_{k'_B}}(\alpha (b'_{k'_B-1})*_{B_{k'_B -1}} (\cdots *_{B_3}(\alpha (b'_2) *_{B_2}(\alpha (b'_1)*_{B_1}p_0)) \cdots )) \\
&\stackrel{\alpha \,\, {\rm hom}}{=}& \alpha (b'_{k'_B}*_{B_{k'_B}}(b'_{k'_B-1}*_{B_{k'_B -1}} (\cdots *_{B_3}(b'_2 *_{B_2}(b'_1*_{B_1}a_0)) \cdots )) \\
&=&\alpha (\phi _{(b',o'_B)}(a_0))=\alpha (\beta (a_0))=K_A.
\end{eqnarray*}
\end{proof}

On the other hand, an attack on Alice's secret key requires (first) the solution of the following problem.
\begin{list}{}{\setlength{\itemsep}{0cm} \setlength{\parsep}{0cm} }
\item[{\bf modHomSP} (Modified Homomorphism Search Problem for $S_B$):]
\item[{\sc Input:}]  Element pairs $(t_1,t'_1),\ldots ,(t_n,t'_n)\in L^2$ with $t'_i=\phi _{a, \alpha }(t_i)$ $\forall 1\le i\le n$ for some (unknown) 
magma homomorphism $\phi _{(a, o_A)}\in End(S_B)$ (with $o_A\in O_A^{k_A}$). Furthermore, an element $p_0 \in \phi _{(a,o_A)}(S_A)$, i.e.
$p_0=\phi _{(a, o_A)}(a_0)$ for some $a_0\in S_A$.
\item[{\sc Objective:}] Find $k'_A\in \mathbb{N}$, $(a'_0, \phi _{(a', o'_A)})\in L \times End(S_B)$ ($o'_A\in O_A^{k'_A}$) such that 
$\phi _{(a', o'_A)}(t_i)=t'_i$ for all $i=1,\ldots ,n$ {\it and} $\phi _{(a', o'_A)}(a'_0)=p_0$.
\end{list}

We define the {\it generalized modHomSP} for $(S_B, S_A)$ as a modified Homomorphism Search Problem for $S_B$ 
with the objective to find $(a'_0, \phi _{(a', o'_A})\in S_A \times End(S_B)$ ($o'_A\in O_A^{k'_A}$) such that $\phi _{(a', o'_A)}(t_i)=t'_i$ for all $i=1,\ldots ,n$ {\it and} $\phi _{(a', o'_A)}(a'_0)=p_0$.  
\par Even if an attacker finds a pseudo-key $(a'_0, \phi _{(a', o'_A)}) \in S_A \times End(S_B)$ for Alice's secret, 
then she still faces an $O_A$-submagma Membership Search Problem for $S_A$.

\begin{prop}  \label{BaseProbs2KEP2}
An oracle that solves the generalized modHomSP for 
$(S_B, S_A)$ and $O_A$-MSP for $S_A$ is sufficient to break key establishment Protocol 2.
\end{prop}

\begin{proof}
As outlined above, we perform an attack on Alice's private key. The generalized modHomSP oracle provides a pseudo-key 
$(a'_0, \phi _{(a', o'_A)}) \in S_A \times Hom(S_B)$
such that $\phi _{(a', o'_A)}(t_i)=t'_i=\alpha (t_i)$ for all $i=1,\ldots ,n$ {\it and} 
$\phi _{(a', o'_A)}(a'_0)=p_0$. 
Observe that this implies for any element $e_B\in S_B$ that 
 $\phi _{(a', o'_A)}(e_B)=\alpha (e_B)$.
In particular, we have $\phi _{(a', o'_A)}(b_i)=\alpha (b_i)$ for all $1\le i \le k_B$.
Since $a'_0 \in S_A$, we may feed $a'_0$ into a $O_A$-MSP oracle for $S_A$ which returns a tree-word  
$T'_{O_A}(r_1,\ldots ,r_l)=a'_0$ (for some $l\in \mathbb{N}$ and $r_i \in \{ s_j \} _{j\le m}$). Now, we may compute 
\begin{eqnarray*}
K'_A &=& \phi _{(a', o'_A)}(T'_{O_A}(\beta (r_1),\ldots ,\beta (r_l))  \stackrel{\beta \, {\rm homo}}{=}  
\phi _{(a', o'_A)}(\beta (T'_{O_A}(r_1,\ldots ,r_l)))  \\
&=& \phi _{(a', o'_A)}(\beta (a'_0)) = \phi _{(a',o'_A)}(b_{k_B}*_{B_{k_B}}(\cdots *_{B_3}(b_2*_{B_2}(b_1*_{B_1}a'_0))\cdots )) \\
&\stackrel{{\rm homo}}{=}& \phi _{(a',o'_A)}(b_{k_B})*_{B_{k_B}}(\cdots *_{B_3}(\phi _{(a',o'_A)}(b_2)*_{B_2}(\phi _{(a',o'_A)}(b_1)*_{B_1}\phi _{(a',o'_A)}(a'_0)))\cdots ) \\
&=&   \alpha (b_{k_B})*_{B_{k_B}}(\cdots *_{B_3}(\alpha (b_2)*_{B_2}(\alpha (b_1)*_{B_1}p_0))\cdots )=K_B.
\end{eqnarray*}
\end{proof}

Now, we describe approaches to break Protocol 2 which do not resort to solving a submagma-MSP. 

\begin{prop} \label{BaseProbs3KEP2} 
A generalized HomSP oracle is sufficient to break key establishment Protocol 2.
More precisely, an oracle that solves the generalized HomSP for $(S_A, S_B)$ {\it and} the HomP for $S_B$ is sufficient to break KEP 2.
\end{prop}

\begin{proof}
Here we perform attacks on Alice's and Bob's private keys - though we do not require a pseudo-key for the first component $a_0$ of Alice's key. 
The HomSP oracle for $S_B$ provides $\phi _{(a', o'_A)}$ with $a'\in L^{k'_A}$ and $o'_A\in O_A^{k'_A}$ s.t. $\phi _{(a', o'_A)}(t_j)=t'_j=\alpha (t_j)$ for all $j\le n$. 
And the generalized HomSP oracle for $(S_A,S_B)$ returns the pseudo-key endomorphism $\phi _{(b', o'_B)}$ with $b'\in S_B^{k'_B}$ and $o'_B\in O_B^{k'_B}$
s.t. $\phi _{(b', o'_B)}(s_i)=s'_i=\beta (s_i)$ for all $i\le m$. Since $b'\in S_B^{k'_B}$, 
we conclude that $\phi _{(a', o'_A)}(b'_i)=\alpha (b'_i)$ for all $1\le i \le k'_B$.
Also, $a_0\in S_A$ implies, of course, $\phi _{(b', o'_B)}(a_0)= \beta (a_0)$.
Let $o'_B=(*_{B'_1}, \ldots , *_{B'_{k'_B}})$
Now, we compute 
\begin{eqnarray*} 
K'_B &=& \phi _{(a',o'_A)}(b'_{k'_B})*_{B'_{k'_B}}(\cdots *_{B'_3}(\phi _{(a',o'_A)}(b'_2)*_{B'_2}(\phi _{(a',o'_A)}(b'_1)*_{B'_1} p_0))\cdots )  \\
&=& \alpha (b'_{k'_B})*_{B'_{k'_B}}(\cdots *_{B'_3}(\alpha (b'_2)*_{B'_2}(\alpha (b'_1)*_{B'_1} \alpha (a_0)))\cdots ) \\
&\stackrel{{\rm homo}}{=}& \alpha (b'_{k'_B}*_{B'_{k'_B}}(\cdots *_{B'_3}(b'_2*_{B'_2}(b'_1*_{B'_1} a_0))\cdots ))  \\
&=& \alpha (\phi _{(b',o'_B)}(a_0))=\alpha (\beta (a_0))=K_A.
\end{eqnarray*}
\end{proof}

Alternatively, one may choose the following approach.
\begin{prop}  \label{BaseProbs4KEP2} 
An oracle that solves the generalized modHomSP for
$(S_B,S_A)$ {\it and} the HomSP for $S_A$ is sufficient to break KEP 2.
\end{prop} 

\begin{proof}
Also here we perform attacks on Alice's and Bob's private keys.  
The HomSP oracle for $S_A$ provides $\phi _{(b',o'_B)}\in End(S_A)$ (with $b'\in S_B^{k'_B}$ and $o'_B=(*_{B'_1}, \ldots ,*_{B'_{k'_B}})\in O_B^{k'_B}$) s.t. $\phi _{(b',o'_B)}(s_j)=s'_j=\beta (s_j)$ for all $j\le m$. 
And the generalized modHomSP oracle for $(S_B,S_A)$ returns the pseudo-key $(a'_0, \phi _{(a',o'_A)}) \in S_A \times End(S_B)$ (with $a'\in S_A^{k'_A}$ and $o'_A\in O_A^{k'_A}$)
s.t. $\phi _{(a',o'_A)}(t_i)=t'_i=\alpha (t_i)$ for all $i\le n$ {\it and} $\phi _{(a',o'_A)}(a'_0)=p_0$. Since $a'_0\in S_A$, we conclude that 
$\phi _{(b',o'_B)}(a'_0)=\beta (a'_0)$.
Also, $b\in S_B^{k_B}$ implies, of course, $\phi _{(a',o'_A)}(b_i)=\alpha (b_i)$ for all $1\le i \le k_B$.
Now, we compute 
\begin{eqnarray*}
K'_A &=& \phi _{(a',o'_A)}( \phi _{(b',o'_B)}( a'_0)) =\phi _{(a',o'_A)}(\beta (a'_0)) \\
&=& \phi _{(a',o'_A)}(b_{k_B}*_{B_{k_B}}(\cdots *_{B_3}(b_2*_{B_2}(b_1*_{B_1} a'_0))\cdots ))  \\
&\stackrel{{\rm homo}}{=}&  
\phi _{(a',o'_A)}(b_{k_B})*_{B_{k_B}}(\cdots *_{B_3}(\phi _{(a',o'_A)}(b_2)*_{B_2}(\phi _{(a',o'_A)}(b_1)*_{B_1} \phi _{(a',o'_A)}(a'_0)))\cdots )  \\
&=& \alpha (b_{k_B})*_{B_{k_B}}(\cdots *_{B_3}(\alpha (b_2)*_{B_2}(\alpha (b_1)*_{B_1} p_0))\cdots )=K_B.
\end{eqnarray*}
\end{proof}

\section{Instantiations using shifted conjugacy}

\subsection{Protocol 1}
Consider the infinite braid group $(B_{\infty },*)$ with shifted conjugacy as LD-operation. Then the iterated LD-Problem is a
\emph{simultaneous iterated shifted conjugacy problem}.
For $m=n=1$ this becomes an \emph{iterated shifted conjugacy problem}.
The \emph{shifted conjugacy problem}. (see e.g. \cite{De06}) which was first solved in \cite{KLT09}
by a double reduction, first to the subgroup conjugacy problem for $B_{n-1}$ in $B_n$, then to an instance of the simultaneous conjugacy problem. 
For the simultaneous conjugacy problem in braid groups we refer to \cite{LL02,KT13}.
As the shifted CP, also the iterated shifted CP can be reduced to a  subgroup conjugacy problem for $B_{n-1}$ in $B_n$.
Even if we replace shifted conjugacy by generalized shifted conjugacy, then the corresponding iterated LD-problem still reduces to a subgroup conjugacy problem for
a standard parabolic subgroup of a braid group.
Such problems were first solved in a more general framework, namely for Garside subgroups of Garside groups, in \cite{KLT10}. 
Though not explicitly stated in \cite{KLT09, KLT10}, the simultaneous shifted conjugacy problem and its analogue for generalized shifted conjugacy
may be treated by similar methods as in \cite{KLT09, KLT10}.
Though these solutions provide only deterministic algorithms with exponential worst case complexity,
they may still affect the security of Protocol 1 if we use such LD-systems in braid groups as platform LD-systems. 
Moreover, efficient heuristic approaches to the shifted conjugacy problem were developed in \cite{LU08,LU09}.
Therefore, one may doubt whether an instantiation of Protocol 1 using (iterated) shifted conjugacy in braid groups provides a secure KEP.
Nevertheless, it is still more interesting than the classical AAG-KEP for braid groups, and it 
might be considered as a first challenge for an possible attacker.

\subsection{Protocol 2} \label{P2}
Here we propose a natural instantiation of Protocol 2 using generalized shifted conjugacy in braid groups.
Consider the following natural partial multi-LD-system $(B_{\infty }, O_A \cup O_B)$ in braid groups. 
\par 
Let $1<q_1<q_2<p$ such that $q_1, p-q_2 \ge 3$. 
Let any $*_{\alpha }\in O_A$ be of the form
$x*_{\alpha }y= \partial ^p(x^{-1}) \alpha \partial ^p(y)x$ with $\alpha =\alpha _1 \tau _{p,p}^{\pm 1} \alpha _2$ for some $\alpha _1 \in B_{q_1}$,
$\alpha _2 \in B_{q_2}$.
Analogously, any $*_{\beta }\in O_B$ is of the form $x*_{\beta }y= \partial ^p(x^{-1}) \beta \partial ^p(y)x$ with 
$\beta = \beta _1 \tau _{p,p}^{\pm 1} \beta _2$ for some $\beta _1  \partial ^{q_2} \in (B_{p-q_2})$, $\beta _2 \in \partial ^{q_1}(B_{p-q_1})$.
Since $[\alpha _1, \beta _1]=[\alpha _1, \beta _2]=[\beta _1, \alpha _2]=1$, 
$(B_{\infty }, *_{\alpha } , *_{\beta })$ is a mutual left distributive system according to Proposition \ref{abc} (d). 
Note that, if in addition we have $[\alpha _1, \alpha _2]=[\beta _1, \beta _2]=1$, then 
$(B_{\infty }, *_{\alpha } , *_{\beta })$ is a bi-LD-system according to Proposition \ref{abc} (c).
But in general these additional commutativity relations do not hold for our choice of standard parabolic subgroups as domains for 
$\alpha _1, \alpha _2, \beta _1, \beta _2$. Note that, if we restrict $\alpha _2, \beta _2$ to $\partial ^{q_1}(B_{q_2-q_1})$, 
then these additional relations are enforced.
Anyway, they are not necessary for $(B_{\infty }, *_{\alpha } , *_{\beta })$ being a mutual left distributive system. 
In either case, $\alpha _2$ does not need to commute with $\beta _2$.  
\par  
Alice and Bob perform the protocol steps of Protocol 2 for the partial multi-LD-system $(B_{\infty }, O_A \cup O_B)$ 
as described in section \ref{KEPpartial}. 
\par 
The deterministic algorithms from \cite{KLT09, KLT10} do not affect the security of this instantiation of Protocol 2, 
because the \emph{operations are part of the secret}.
\par
We provide an explicit formula for the public information (here $s'_i$) depending on $s_i$ and Bob's secret keys, namely $k=k_B\in \mathbb{N}$ and $(b,o_{\beta })\in S_B^k\times O_B^k$ where $b=(b_1, \ldots , b_k)$ and $o_{\beta }=(*_{\beta _1}, \ldots ,*_{\beta _k})$ and 
$x*_{\beta _i}=\beta '_i \tau _{p,p}^{\epsilon _i} \beta ''_i$ with $\beta '  \in \partial ^{q_2}(B_{p-q_2}) \subseteq B_p$,
$\beta ' \in \partial ^{q_1}(B_{p-q_1}) \subseteq B_p$ and $\epsilon _i \in \{\pm 1 \}$.  
Let $\tilde{b}=\partial ^{(k-1)p}(b_1) \cdots \partial ^p(b_{k-1})b_k$, $\tilde{\beta }'=\beta '_k \beta '_{k-1} \cdots \beta '_1$
and $\tilde{\beta }''=\prod _{i=1}^{k} \partial ^{(i-1)p}(\beta ''_i)$. Then we have
\begin{equation}
s'_i=\partial ^p(\tilde{b}^{-1}) \tilde{\beta }' \partial ^p(\tilde{\beta }'') 
\tau _{p,p}^{\epsilon _k} \partial ^p(\tau _{p,p}^{\epsilon _{k-1}}) \cdots \partial ^{(k-1)p}(\tau _{p,p}^{\epsilon _1}) 
\partial ^{kp}(s_i) \tilde{b}.  
\end{equation}

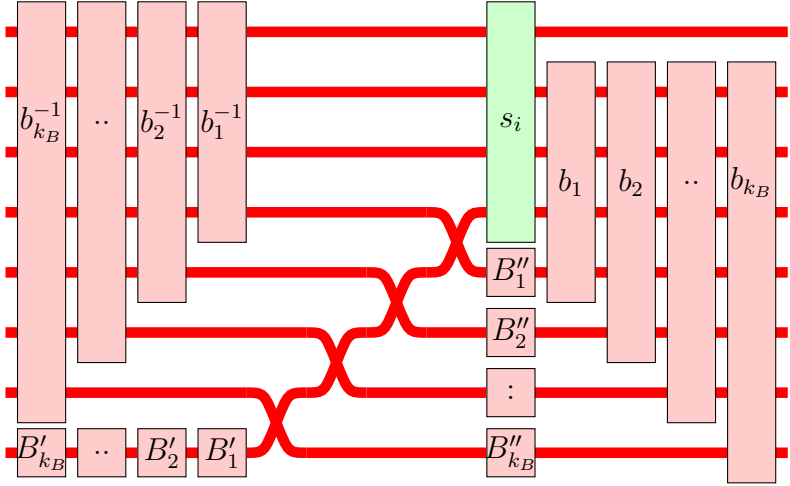
\begin{figure}[ht]
  \caption{Structure of Bob's public key $s'_i$.}
\begin{center}
\begin{tikzpicture}  [scale = 0.8]
\draw [line width = 4pt, red] (0,0) -- ++(4,0);
\draw [line width = 4pt, red] (0,1) -- ++(4,0);
\draw [line width = 4pt, red] (0,2) -- ++(5,0);
\draw [line width = 4pt, red] (0,3) -- ++(6,0);
\draw [line width = 4pt, red] (0,4) -- ++(7,0);
\draw [line width = 4pt, red] (0,5) -- ++(13,0);
\draw [line width = 4pt, red] (0,6) -- ++(13,0);
\draw [line width = 4pt, red] (0,7) -- ++(13,0);

\draw [line width = 4pt, red] (5,0) -- ++(8,0);
\draw [line width = 4pt, red] (6,1) -- ++(7,0);
\draw [line width = 4pt, red] (7,2) -- ++(6,0);
\draw [line width = 4pt, red] (8,3) -- ++(5,0);
\draw [line width = 4pt, red] (8,4) -- ++(5,0);

\draw [line width = 4pt, red] (4,0) .. controls +(0.3,0)  .. +(0.5,0.5) .. controls +(0.7,1) .. +(1,1);
\draw [line width = 4pt, red] (4,1) .. controls +(0.3,0)  .. +(0.5,-0.5) .. controls +(0.7,-1) .. +(1,-1);
\draw [line width = 4pt, red] (5,1) .. controls +(0.3,0)  .. +(0.5,0.5) .. controls +(0.7,1) .. +(1,1);
\draw [line width = 4pt, red] (5,2) .. controls +(0.3,0)  .. +(0.5,-0.5) .. controls +(0.7,-1) .. +(1,-1);
\draw [line width = 4pt, red] (6,2) .. controls +(0.3,0)  .. +(0.5,0.5) .. controls +(0.7,1) .. +(1,1);
\draw [line width = 4pt, red] (6,3) .. controls +(0.3,0)  .. +(0.5,-0.5) .. controls +(0.7,-1) .. +(1,-1);
\draw [line width = 4pt, red] (7,3) .. controls +(0.3,0)  .. +(0.5,0.5) .. controls +(0.7,1) .. +(1,1);
\draw [line width = 4pt, red] (7,4) .. controls +(0.3,0)  .. +(0.5,-0.5) .. controls +(0.7,-1) .. +(1,-1);

\fill [red!20] (0.2,0.5) -- ++(0.8,0) -- ++(0,7) -- ++(-0.8,0) -- cycle;
\draw (0.2,0.5) -- ++(0.8,0) -- ++(0,7) -- ++(-0.8,0) -- cycle;
\node at (0.6,5.5) {$b_{k_B}^{-1}$};
\fill [red!20] (1.2,1.5) -- ++(0.8,0) -- ++(0,6) -- ++(-0.8,0) -- cycle;
\draw (1.2,1.5) -- ++(0.8,0) -- ++(0,6) -- ++(-0.8,0) -- cycle;
\node at (1.6,5.5) {..};
\fill [red!20] (2.2,2.5) -- ++(0.8,0) -- ++(0,5) -- ++(-0.8,0) -- cycle;
\draw (2.2,2.5) -- ++(0.8,0) -- ++(0,5) -- ++(-0.8,0) -- cycle;
\node at (2.6,5.5) {$b_2^{-1}$};
\fill [red!20] (3.2,3.5) -- ++(0.8,0) -- ++(0,4) -- ++(-0.8,0) -- cycle;
\draw (3.2,3.5) -- ++(0.8,0) -- ++(0,4) -- ++(-0.8,0) -- cycle;
\node at (3.6,5.5) {$b_1^{-1}$};

\fill [red!20] (12,-0.5) -- ++(0.8,0) -- ++(0,7) -- ++(-0.8,0) -- cycle;
\draw (12,-0.5) -- ++(0.8,0) -- ++(0,7) -- ++(-0.8,0) -- cycle;
\node at (12.4,4.5) {$b_{k_B}$};
\fill [red!20] (11,0.5) -- ++(0.8,0) -- ++(0,6) -- ++(-0.8,0) -- cycle;
\draw (11,0.5) -- ++(0.8,0) -- ++(0,6) -- ++(-0.8,0) -- cycle;
\node at (11.4,4.5) {..};
\fill [red!20] (10,1.5) -- ++(0.8,0) -- ++(0,5) -- ++(-0.8,0) -- cycle;
\draw (10,1.5) -- ++(0.8,0) -- ++(0,5) -- ++(-0.8,0) -- cycle;
\node at (10.4,4.5) {$b_2$};
\fill [red!20] (9,2.5) -- ++(0.8,0) -- ++(0,4) -- ++(-0.8,0) -- cycle;
\draw (9,2.5) -- ++(0.8,0) -- ++(0,4) -- ++(-0.8,0) -- cycle;
\node at (9.4,4.5) {$b_1$};

\fill [green!20] (8,3.5) -- ++(0.8,0) -- ++(0,4) -- ++(-0.8,0) -- cycle;
\draw (8,3.5) -- ++(0.8,0) -- ++(0,4) -- ++(-0.8,0) -- cycle;
\node at (8.4,5.5) {$s_i$};

\fill [red!20] (0.2,-0.4) -- ++(0.8,0) -- ++(0,0.8) -- ++(-0.8,0) -- cycle;
\draw (0.2,-0.4) -- ++(0.8,0) -- ++(0,0.8) -- ++(-0.8,0) -- cycle;
\node at (0.6,0) {$B'_{k_B}$};
\fill [red!20] (1.2,-0.4) -- ++(0.8,0) -- ++(0,0.8) -- ++(-0.8,0) -- cycle;
\draw (1.2,-0.4) -- ++(0.8,0) -- ++(0,0.8) -- ++(-0.8,0) -- cycle;
\node at (1.6,0) {..};
\fill [red!20] (2.2,-0.4) -- ++(0.8,0) -- ++(0,0.8) -- ++(-0.8,0) -- cycle;
\draw (2.2,-0.4) -- ++(0.8,0) -- ++(0,0.8) -- ++(-0.8,0) -- cycle;
\node at (2.6,0) {$B'_2$};
\fill [red!20] (3.2,-0.4) -- ++(0.8,0) -- ++(0,0.8) -- ++(-0.8,0) -- cycle;
\draw (3.2,-0.4) -- ++(0.8,0) -- ++(0,0.8) -- ++(-0.8,0) -- cycle;
\node at (3.6,0) {$B'_1$};

\fill [red!20] (8,2.6) -- ++(0.8,0) -- ++(0,0.8) -- ++(-0.8,0) -- cycle;
\draw (8,2.6) -- ++(0.8,0) -- ++(0,0.8) -- ++(-0.8,0) -- cycle;
\node at (8.4,3) {$B''_1$};
\fill [red!20] (8,1.6) -- ++(0.8,0) -- ++(0,0.8) -- ++(-0.8,0) -- cycle;
\draw (8,1.6) -- ++(0.8,0) -- ++(0,0.8) -- ++(-0.8,0) -- cycle;
\node at (8.4,2) {$B''_2$};
\fill [red!20] (8,0.6) -- ++(0.8,0) -- ++(0,0.8) -- ++(-0.8,0) -- cycle;
\draw (8,0.6) -- ++(0.8,0) -- ++(0,0.8) -- ++(-0.8,0) -- cycle;
\node at (8.4,1) {:};
\fill [red!20] (8,-0.4) -- ++(0.8,0) -- ++(0,0.8) -- ++(-0.8,0) -- cycle;
\draw (8,-0.4) -- ++(0.8,0) -- ++(0,0.8) -- ++(-0.8,0) -- cycle;
\node at (8.4,0) {$B''_{k_B}$};
\end{tikzpicture}
\end{center}
\end{figure}

For $\epsilon =(\epsilon _1, \ldots , \epsilon _k)\in \{ \pm 1\}^k$, we introduce the abreviation
\[ \tau (p, \epsilon )=
\tau _{p,p}^{\epsilon _k} \partial ^p(\tau _{p,p}^{\epsilon _{k-1}}) \cdots \partial ^{(k-1)p}(\tau _{p,p}^{\epsilon _1}). \]
We conclude that the Homomorphism Search Problem for $S_A$ specifies to the following particular (modified)
simultaneous decomposition problem.
\begin{list}{}{\setlength{\itemsep}{0cm} \setlength{\parsep}{0cm} }
\item[{\sc Input:}]  Element pairs $(s_1,s'_1),\ldots ,(s_m,s'_m)\in B_{\infty }^2$ with 
$$ s'_i=\partial ^p(\tilde{b}^{-1}) \tilde{\beta }' \partial ^p(\tilde{\beta }'') \tau (p,\epsilon ) \partial ^{kp}(s_i) \tilde{b}$$
for all $i$, $1\le i\le m$, for some (unknown) $k\in \mathbb{N}$, $\epsilon \in \{\pm \}^k$, $\tilde{b} \in B_{\infty }$, 
$$\tilde{\beta }' \in \partial ^{q_2}(B_{p-q_2}), \quad \tilde{\beta }'' \in \prod _{j=1}^{k} \partial ^{(j-1)p+q_1}(B_{p-q_1}). $$  
\item[{\sc Objective:}] Find $k' \in \mathbb{N}$, $\epsilon ' \in \{\pm \}^{k'}$,
$\hat{b} \in B_{\infty }$, $\hat{\beta }' \in \partial ^{q_2}(B_{p-q_2})$, $\hat{\beta }'' \in \prod _{j=1}^{k'} \partial ^{(j-1)p+q_1}(B_{p-q_1})$
such that $$s'_i=\partial ^p(\hat{b}^{-1}) \hat{\beta }' \partial ^p(\hat{\beta }'') \tau (p, \epsilon ') \partial ^{k'p}(s_i) \hat{b}$$ for all $i=1,\ldots ,m$.
\end{list}

Note that one has also to determine the iteration depth $k=k_B$ of Bob's secret homomorphism $\beta $ (or some pseudo iteration depth $k'$)
as well the bit sequence $\epsilon \in \{\pm 1\}^k$.
Since all instance elements live in some $B_N \subseteq B_{\infty }$ for some finite $N\in \mathbb{N}$, it is easy to obtain an upper bound for $k$ from $N$. 
\begin{rem} 
If we abandon simultaneity, i.e. in the case $m=1$, we obtain a \emph{(modified) special decomposition problem}.
In the following section we transform this particular problem to finitely many instances of the subgroup conjugacy coset problem for
parabolic subgroups of braid groups.
\end{rem}

\subsection{Conjugacy coset problem}

\begin{defi} \label{SCCP}
Let $H, K$ be subgroups of a group $G$. We call the following problem the subgroup conjugacy coset problem (SCCP) for $(H,K)$ in $G$.
\begin{list}{}{\setlength{\itemsep}{0cm} \setlength{\parsep}{0cm} }
\item[{\sc Input:}]  An element pair $(x,y) \in G^2$ such that $x^G \cap Hy \ne \emptyset $.
\item[{\sc Objective:}]  Find elements $h \in H$ and $c \in K$ such that $cxc^{-1}=hy$.
\end{list}
If $K=G$ then we call this problem the conjugacy coset problem (CCP) for $H$ in $G$.
\end{defi}

This is the search (or witness) version of this problem. The corresponding decision problem is to decide whether
the conjugacy class of $x$ and the left $H$-coset of $y$ intersect, i.e. whether $x^G \cap Hy \stackrel{?}{=} \emptyset $.
Anyway, in our cryptographic context we usually deal with search problems. 
\par 
It is clear from the definition that the SCCP is harder than the double coset problem (DCP) and the subgroup conjugacy problem (subCP), i.e., 
an oracle that solves SCCP for any pair $(H,K)\le G^2$ also solves DCP and subCP.  
\par 
Though the CCP and the SCCP are natural group-theoretic problems, they seem to have attracted little attention in combinatorial group theory so far.
To our knowledge they have been introduced in \cite{KT13}.  
\par 
We connect the (modified) special decomposition problem from the previous section to the SCCP.

\begin{prop} \label{SDP}
The (modified) special decomposition problem (for $m=1$) from section \ref{P2} can be solved by solving $2^k$ instances of the SCCP for some standard parabolic subgroups in braid groups,
namely the SCCP for $( \partial ^{N-p+q_2}(B_{p-q_2}) \cdot \prod _{j=1}^k \partial ^{(j-1)p+q_1}(B_{p-q_1}), B_{N-p})$ in $B_N$ for some $k, N \in \mathbb{N}$.
\end{prop}

\begin{proof}
For $m=1$, we write $s=s_m$ and $s'=s'_m$.
Let $N \in \mathbb{N}$ be sufficiently large such that $s' \in B_N$. For convenience, we choose a minimal $N$ such that $N\ge (k+1)p$ and $p\mid N$.
As in \cite{KLT09, KT13} we conclude that $\tilde{b}\in B_{N-p}$ and $\partial ^p(\tilde{b}^{-1}) \in \partial ^p (B_{N-p})$. Therefore we have
$$\tau _{p,N-p}^{-1} \partial ^p(\tilde{b}^{-1})=\tilde{b}^{-1}\tau _{p,N-p}^{-1}. $$
Furthermore, since $\tau _{p,\epsilon }\tilde{\beta }''=\partial ^p(\tilde{\beta }'') \tau _{p,\epsilon}$, 
we get
\[ \begin{array}{rcll}
s' &=& \partial ^p(\tilde{b}^{-1}) \tilde{\beta }' \partial ^p(\tilde{\beta }'') \tau (p,\epsilon ) \partial ^{kp}(s) \tilde{b} & \Leftrightarrow \\
\tau _{p,N-p}^{-1}s' &=& \tilde{b}^{-1} \tau _{p,N-p}^{-1}  \tilde{\beta }' \partial ^p(\tilde{\beta }'') \tau (p,\epsilon ) \partial ^{kp}(s) \tilde{b} & \\
 &=&   \tilde{b}^{-1} \partial ^{N-p} ( \tilde{\beta }') \tilde{\beta }'' \tau _{p,N-p}^{-1}\tau (p,\epsilon ) \partial ^{kp}(s)  \tilde{b} & \Leftrightarrow \\
  \tilde{b}\tilde{s}'  \tilde{b}^{-1} &=& \tilde{\beta } \cdot \tilde{s} &
\end{array} \]
with $\tilde{s}'=\tau _{p,N-p}^{-1}s'$, $\tilde{s}=\tau _{p,N-p}^{-1}\tau _{p,\epsilon } \partial ^{kp} (s)$, and 
$$\tilde{\beta }=\partial ^{N-p} ( \tilde{\beta }')  \tilde{\beta }'' 
\in  \partial ^{N-p+q_2}(B_{p-q_2})  \cdot  \prod _{j=1}^{k} \partial ^{(j-1)p+q_1}(B_{p-q_1}).$$
So, if we solve this SCCP for all $\epsilon \in \{\pm 1\}^k$, we obtain a solution to the (modified) special decomposition problem (for $m=1$) from section \ref{P2}. Note that $|\{ \pm 1\}^k|=2^k$.
\end{proof}

\par Recall that the algorithms from \cite{KLT09, KLT10}, as well as from \cite{GKLT13}, only solve
instances of the subgroup conjugacy problem for parabolic subgroups of braid groups, partially by transformation to the simultaneous conjugacy problem
in braid groups \cite{KTTV13}.
No deterministic or even heuristic solution to the SCCP for (standard) parabolic subgroups in braid groups is known yet. 

\medskip

{\bf Open problem.} Find a solution to the SCCP, or even the CCP, for (standard) parabolic subgroups in the braid group $B_N$. 

\medskip

The attacker might try to approach first an apparently much easier (but still open) problem, namely the SCCP,  or even the CCP, for (standard) parabolic subgroups in the symmetric group $S_N$, which is a natural qoutient of $B_N$, given by the homomorphism $\sigma _i \mapsto (i,i+1)$. 
\par
The CCP (and the SCCP) appear to be inherently quad\-ratic, i.e. we do not see how it may be linearized such that linear algebra attacks
as the \emph{linear centralizer attack} of B. Tsaban \cite{Ts12} apply. It shares this feature with Y. Kurt's {\it Triple Decomposition Problem} 
(see section 4.2.5. in \cite{MSU11}).  
\medskip
Note that $k$ is still unknown to the attacker, but $N$ (even $N/p$) is surely an upper bound for $k$. Anyway, it suffices to solve $O(2^N)$ SCCP-instances.
This is the main advantage of the iterated Protocol 2 compared to Protocol 2 from \cite{KT13} (not iterated).
\par
\emph{Remark.} But the SCCP for $( \partial ^{N-p+q_2}(B_{p-q_2}) \cdot \prod _{j=1}^k \partial ^{(j-1)p+q_1}(B_{p-q_1}), B_{N-p})$ in $B_N$ 
admits a small disadvantage - compared to a SCCP for $(H,K)$ in $B_N$ for arbitrary (standard) parabolic subgroups $H, K$ of $B_N$ - 
which hasn't been pointed out in \cite{KT13}. \\
$\tilde{s}$ lives in $\partial ^{kp} (B_{N-kp})$. Therefore, $\partial ^{N-p} ( \tilde{\beta }') \cdot \tilde{s}$ commutes with 
\[ \tilde{\beta }'' \in  \prod _{j=1}^k \partial ^{(j-1)p+q_1}(B_{p-q_1}) \subseteq B_{kp}. \] Therefore, the attacker may conclude that $\tilde{s}'$ is conjugated (by a conjugator $\tilde{b} \in B_{N-p}$) to an element in the (standard) parabolic subgroup 
$\partial ^{kp}(B_{N-kp}) \cdot \prod _{j=1}^k \partial ^{(j-1)p+q_1}(B_{p-q_1}), B_{N-p})$, 
namely $\partial ^{N-p} ( \tilde{\beta }') \tilde{s} \cdot \tilde{\beta }''$. Using Nielsen-Thurston theory or some kind of subgroup distance attack
this feature might be exploited by the attacker. We leave this as an open problem.
 
\subsection{Challenges} 
As a challenge for a possible attacker, we provide some suggestions for the involved parameter values.  
\par
(1) Since the complexity of the involved braids might grow exponentially with the number $l$ of internal nodes of the involved p.r.b. trees, an 
implementation of Protocol 2 in braid groups (as outlined in section \ref{P2}) can only be efficient for small parameter values.
 Nevertheless, as a challenge, we suggest, for example, the following parameter values. 
We abandon simultaneity, i.e. we set $m=n=1$. The braids $s_1,t_1, a_1, \ldots , a_{k_A}, b_1, \ldots , b_{k_B}$ are generated as "random" signed words
(over the standard generators $\sigma _i$) of length $L=15$ in $B_N$ with $N=10$.
We choose $p=6$ for the generalized shift and $q=q_1=q_2=p/2=3$.
The braid $A'_1, \ldots ,A'_{k_A}$ and $A''_1, \ldots , A''_{k_A}$ are generated as "random" signed words
(over the standard generators $\sigma _i$) of length $L_{ops}=5$ in $B_q$. The $B'_i$'s and $B''_i$'s are chosen analogeoulsy, but from $\partial ^q(B_q)$.
The iteration depths are set to $k_A=k_B=3$, and we set the number $l=l_A=l_B$ of internal nodes of the involved planar rooted binary trees to 4.
\par 
(2) A more efficient implementation in braid groups  can be obtained by using the bi-LD-system $(B_{\infty }, *, \bar{*})$.
As a challenge, we suggest, for example, the following parameter values.
We abandon simultaneity, i.e. we set $m=n=1$. The braids $s_1,t_1, a_1, \ldots , a_{k_A}, b_1, \ldots , b_{k_B}$ are generated as "random" signed words
(over the standard generators $\sigma _i$) of length $L=25$ in $B_N$ with $N=4$.
The iteration depths are set to $k:=k_A=k_B=5$, and we set the number $l=l_A=l_B$ of internal nodes of the involved planar rooted binary trees to 5.
\par
A disadvantage of this scheme is the following.
Analogeously to Proposition \ref{SDP} one may attack this scheme by solving $2^k$ instances of the subgroup CP for some standard parabolic subgroup in braid groups, namely the subgroup CP for $B_{N-1}$ in $B_N$ for some $N \in \mathbb{N}$.
\par
(3) An extremely efficient implementation of Protocol 2 (as outlined in section \ref{P2}) can be obtained by working in in the quotient system 
$S_{\infty }$ rather than the partial multi-LD-system $B_{\infty }$. 
Here we may choose much larger parameter values as a challenge.
We abandon simultaneity, i.e. we set $m=n=1$. The permutations $s_1,t_1, a_1, \ldots , a_{k_A}, b_1, \ldots , b_{k_B}$ are generated as random
permutations in $S_N$ with $N=200$.
We choose $p=20$ for the generalized shift and $q=q_1=q_2=p/2=10$.
The  permutations $A'_1, \ldots ,A'_{k_A}$ and $A''_1, \ldots , A''_{k_A}$ are generated as random
permutations in $S_q$.
The iteration depths $k_A, k_B$ are chosen from the interval $[2, \ldots , 30]$, 
and we choose the numbers of internal nodes of the involved planar rooted binary trees from the interval $[10, \dots , 20]$.
\par
Analogeously to Proposition \ref{SDP} one may attack Bob's secret by  solving $O(k)$ ($k=k_B$) instances of the SCCP for some standard parabolic subgroups in symmetric groups, namely the SCCP for $( \partial ^{N-p+q_2}(S_{p-q_2}) \cdot \prod _{j=1}^k \partial ^{(j-1)p+q_1}(S_{p-q_1}), S_{N-p})$ in $S_N$ for some $k, N \in \mathbb{N}$.
Note that here for an attack on Bob's key the solution $O(k)$ (rather than $2^k$) SCCP-instance suffices.
This is because under the surjection $B_{\infty } \rightarrow S_{\infty}$, $\tau (p, \epsilon )$ maps to  the fixed permutation  
\[ \left( \begin{array}{cccccc} 1    & \cdots & p      & p+1 & \cdots & (k+1)p \\
                            kp+1 & \cdots & (k+1)p & 1   & \cdots & kp \end{array} \right)\]
 for all $\epsilon \in \{\pm 1\}^k$, and only $k$ remains unknown.
\par
\medskip

\section{Other instantiations}

\subsection{Instantiations using $f$-conjugacy}

A straightforward computation yields the following proposition.
\begin{prop}
Let $G$ be a group and $f_1, f_2 \in End(G)$. Then $(G, *_{f_1}, *_{f_2})$ with $x*_{f_i}y=f_i(x^{-1}y)x$ (for $i=1,2$) is a mutually left distributive
system if and only if $f_1=f_2$.
\end{prop}
Therefore, we don't have any nontrivial partial multi-LD-structures using $f$-conjugacy. We only have the platform LD-system $(G,*_f)$ for some fixed
endomorphism $f\in End(G)$ and we can only apply Protocol 1.
\par
In Protocol 1 Bob's public key consist of elements $s'_i=b_{k_B}*_f( \cdots b_2 *_f(b_1*_f s_i) \cdots )$ (for $i=1, \ldots ,m$).
Evaluating the right hand side, we obtain
\begin{eqnarray*}
 s'_i &=& f(b_{k_B}^{-1}) \cdots f^{k_B-1}(b_2^{-1})f^{k_B}(b_1^{-1})\cdot f^{k_B}(s_i) \cdot f^{k_B-1}(b_1) \cdots f(b_{k_B-1})b_{k_B}. \\
   &=& f(\tilde{b}) f^{k_B}(s_i) \tilde{b} \quad {\rm with} \quad \tilde{b}=f^{k_B-1}(b_1) \cdots f(b_{k_B-1})b_{k_B}.
\end{eqnarray*}
Since $\langle s_1 \rangle _{*_f}=\{ s_1\}$, we cannot abandon simultaneity for $f$-conjugacy, i.e. we have $m\ge 2$.
Therefore, we have for $1\le i\ne j\le m$
\begin{eqnarray*}
s'_i(s'_j)^{-1} &=& f(\tilde{b}^{-1}) f^{k_B}(s_is_j^{-1}) f(\tilde{b}), \quad {\rm and} \\
(s'_j)^{-1}s_i &=& \tilde{b}^{-1} f^{k_B}(s_j^{-1}s_i) \tilde{b}.
\end{eqnarray*}
Now, an attacker might try to solve (in parallel) the following $2U_B$ ${m\choose 2}$-simultaneous CP-instances:
\begin{eqnarray*}
 \{ (s'_i(s'_j)^{-1},f^{k}(s_is_j^{-1})) \mid 1 \le i \ne j\le m \}, \quad \forall k=1, \ldots , U_B, \quad {\rm and} \\
 \{((s'_j)^{-1}s_i, f^{k}(s_j^{-1}s_i)) \mid 1 \le i\ne j \le m \}, \quad \forall k=1, \ldots , U_B.
\end{eqnarray*}
where $U$ denotes some upper bound on $k_B$ which might be obtained from the public keys or parameter specifications of the particular
$f$-conjugacy KEP instantiation. Actually, it suffices to solve the latter $U_B$ ${m\choose 2}$-simultaneous CP-instances.
If the center of $G$ is "small", the attacker might obtain the original private keys $k_B$ and $\tilde{b}$.
Similarly, he might approach Alice's private keys by solving (in parallel) $U_A$ ${n\choose 2}$-simultaneous CP-instances,
where $k_A\le U_A$. Thus she might possibly obtain also $k_A$ and $\tilde{a}=f^{k_A-1}(a_1) \cdots f(a_{k_A-1})a_{k_A}$, and from
these $f^{k_A}(a_0)$. This suffices to recover the shared key
\[ K=f(\tilde{a}^{-1})f^{k_A+1}(\tilde{b}^{-1}) f^{k_A+k_B}(a_0)f^{k_A}(\tilde{b})\tilde{a}. \]
Therefore, it is recommended to choose the generators $s_i$ (and $t_j$) of $S_A$ (and $S_B$) such that the following sets have large centralizers
\[ \{ f^{k_B}(s_is_j^{-1}) \mid 1 \le i \ne j\le m \}, \,\, \{ f^{k_B}(s_is_j^{-1}) \mid 1 \le i \ne j\le m \}, \]
\[ \{ f^{k_A}(t_it_j^{-1}) \mid 1 \le i \ne j\le n \}, \,\, \{ f^{k_A}(t_it_j^{-1}) \mid 1 \le i \ne j\le n \}. \]
Since Alice cannot know $k_B$ (and Bob not $k_A$), this might be achieved by choosing the $s_i$'s and $t_j$'s such that the generator sets of $S_A$ and $S_B$
have already large centralizers. \par

\medskip

\emph{Instantiation in finite matrix groups.} Here we propose an efficient instantiation of the iterated $f$-conjugacy KEP in the finite matrix group
$G=GL(d, \mathbb{F}_{p^N})$ where the $f$-conjugacy operation is given by the homomorphism $f\in End(G)$ that is induced by the Frobenius ring endomorphism
$Fr \in End(\mathbb{F}_{p^N})$, defined by $x \mapsto x^p$. Since $ord (Fr)=N$ induces $ord f=N$, the iteration depths $k_A$, $k_B$ are bounded below $n$.
Therefore, it is recommended to choose $p$ small and $N$ "large". 
As a challenge, we suggest, for example, the following parameter values. 
Set $d=6$, $p=2$, $N=40$, $m=n=8$, and the iteration depths are $k_A=k_B=25$.
We set the number $l=l_A=l_B$ of internal nodes of the involved planar rooted binary trees to 10.

\begin{exa}
As a further example we propose a possible instantiation of the iterated $f$-conjugacy KEP in pure braid groups.
\par 
Recall that the $N$-strand braid group $B_N$ is generated by $\sigma _1$, ..., $\sigma _{N-1}$ where inside $\sigma _i$ the $(i+1)$-th
strand crosses over the $i$-th strand. There exists a natural epimorphism from $B_N$ onto the symmetric group $S_N$, defined by $\sigma _i \mapsto (i,i+1)$.
Let $G$ be the kernel of this epimorphism, namely the $N$-strand pure braid group $P_N$.
For some small integer $d\ge 1$, consider the epimorphism $\eta _d: P_N \longrightarrow P_{N-d}$ given by "pulling out" (or erasing) the last $d$ strands, i.e. the strands $N-d+1, \ldots , N$. Consider the shift map $\partial : B_{N-1} \longrightarrow B_N$, defined by $\sigma _i \mapsto \sigma _{i+1}$,
and note that $\partial ^d (P_{N-d}) \le P_N$. Now, we define the endomorphism $f: P_N \longrightarrow P_N$ by the composition $f=\partial ^d \circ \eta _d$,
and our KEP is Protocol 1 applied  to the LD-system $(P_N, *_f)$. Note that the iteration depths $k_A$, $k_B$ are bound below $N/d$. Here, $d=1$ is of particular
interest since it allows for the biggest upper bound on $k_A$ and $k_B$.
\par 
Alternatively, one may use the following modified scheme. Recall that $P_N$ is generated by the ${N\choose 2}$ elements
\[ A_{i,j}=\sigma _{j-1}\cdots \sigma _{i+1} \sigma _i^2 \sigma _{i+1}^{-1}\cdots \sigma _{j-1}^{-1} \quad (1\le i<j \le N). \]
Now, for $d\ge 3$, define $f\in End(P_N)$ by $f(A_{i,j}^{\pm 1})=c^{\pm 1} \cdot \partial ^d \circ \eta _d(A_{i,j}^{-1})$
where the constant pure braid $c\in P_N$ is given by
$c=\tau _{N-d,d}\tau _{d,N-d} c_0$ for some constant pure braid $c_0\in P_d$.
Inside the pure braid $\tau _{N-d,d}\tau _{d,N-d}$ the first $d$ strands go around the last $N-d$ strands (or vice versa).
\end{exa} 

\begin{rem}
We leave it for future work to construct further instances of the iterated $f$-conjugacy KEP. The following proposition
suggests that any platform LD-system $(G,*_f)$ with $G$ group and $f\in End(G)$ satisfying $f^2\ne f$.
\begin{prop}
Consider the relation $\rightarrow _{*_f}$ induced by $f$-conjugacy, i.e. $x \rightarrow _{*_f}y$ if there exists a $c\in G$ such that
$y=c*_f x=f(c^{-1}x)c$. The relation $\rightarrow _{*_f}$ is transitive if and only if $f$ is a projector, i.e. $f^2=f$.
\end{prop}
Therefore, if $f$ is a projector the iterated $f$-conjugacy KEP doesn't yield any advantage compared to its non-iterated version (see \cite{KT13}).
\par 
Furthermore, the following (iterated version of the) $f$-conjugator search problem should be hard.

\begin{list}{}{\setlength{\itemsep}{0cm} \setlength{\parsep}{0cm} }
\item[{\sc Input:}]  Element pairs $(s_1,s'_1),\ldots ,(s_m,s'_m)\in G^2$ with 
\[ s'_i=f(\tilde{b}^{-1}) f^{k_B}(s_i) \tilde{b} \quad \forall i, \,\, 1\le i\le m, \]
for some (unknown) $k_B\in \mathbb{N}$, $\tilde{b} \in G$.  
\item[{\sc Objective:}] Find $k'_B \in \mathbb{N}$, $\hat{b} \in G$ such that 
\[ s'_i=f(\hat{b}^{-1}) f^{k'_B}(s_i) \hat{b} \quad  \forall i=1,\ldots ,m.  \]
\end{list}
\end{rem}

\subsection{Instantiations using $f$-symmetric conjugacy}
A straightforward computation yields the following proposition.
\begin{prop}
Let $G$ be a group and $f_1, f_2$ two projectors in $End(G)$. Then $(G, *_{f_1}, *_{f_2})$ with $x*_{f_i}y=f_i(xy^{-1})x$ (for $i=1,2$) is a mutually left distributive system if and only if $f_1=f_2$.
\end{prop}
Therefore, we don't have any nontrivial partial multi-LD-structures using $f$-symmetric conjugacy. We only have the platform LD-system $(G,*_f)$ for some fixed
projector $f\in End(G)$ and we can only apply Protocol 1.
\par
In Protocol 1 Bob's public key consist of elements $s'_i=b_{k_B}*_f( \cdots b_2 *_f(b_1*_f s_i) \cdots )$ (for $i=1, \ldots ,m$).
Evaluating the right hand side, we obtain
\[ s'_i= \left\{
\begin{array}{l}
 f(b_{k_B} b_{k_B-1}^{-1} \cdots b_3 b_2^{-1}b_1)\cdot f(s_i^{-1}) \cdot f(b_1b_2^{-1}b_3 \cdots b_{k_B-1}^{-1}) b_{k_B}, \quad k_B \,\, {\rm odd}, \\
 f(b_{k_B}b_{k_B-1}^{-1} \cdots b_3^{-1}b_2b_1^{-1})\cdot f(s_i) \cdot f(b_1^{-1}b_2b_3^{-1} \cdots b_{k_B-1}^{-1}) b_{k_B}, \quad k_B \,\, {\rm even}.
\end{array} \right. \]

Consider the relation $\rightarrow _{*_f}$ induced by $f$-symmetric conjugacy, i.e. $x \rightarrow _{*_f}y$ if there exists a $c\in G$ such that
$y=c*_f x=f(cx^{-1})c$. The relation $\rightarrow _{*_f}$ is never transitive.
Therefore,  the iterated $f$-conjugacy KEP always  provides an advantage compared to its non-iterated version (see \cite{KT13}).
\par
We conclude the following $m$-simultaneous \emph{iterated $f$-symmetric conjugator search problem} should be hard.

\begin{list}{}{\setlength{\itemsep}{0cm} \setlength{\parsep}{0cm} }
\item[{\sc Input:}]  Element pairs $(s_1,s'_1),\ldots ,(s_m,s'_m)\in G^2$ with 
\[ s'_i=f(b_{k_B} \cdots b_3^{\pm 1} b_2^{\mp 1}b_1^{\pm 1}) f(s_i^{\mp 1}) f(b_1^{\pm 1}b_2^{\mp 1}b_3^{\pm 1} \cdots b_{k_B-1}^{-1}) \cdot b_{k_B}  \quad \forall i, \,\, 1\le i\le m, \]
for some (unknown) $k_B\in \mathbb{N}$, $b_1, \ldots , b_k \in G$.  
\item[{\sc Objective:}] Find $k'_B \in \mathbb{N}$, $b'_1, \ldots , b'_{k'_B} \in G$ such that 
\[ s'_i=f(b_{k'_B} \cdots (b'_2)^{\mp 1} (b'_1)^{\pm 1}) f(s_i^{\mp 1}) f((b'_1)^{\pm 1}(b'_2)^{\mp 1} \cdots (b'_{k_B-1})^{-1}) \cdot b'_{k_B} \] for all $i$ with $1\le i\le m$. 
\end{list}

\begin{rem}
As for $f$-conjugacy, since $\langle s_1 \rangle _{*_f}=\{ s_1\}$, we cannot abandon simultaneity for $f$-symmetric conjugacy, i.e. we have $m\ge 2$.
Therefore, we have for $1\le i\ne j\le m$
\begin{eqnarray*}
s'_i(s'_j)^{-1} &=& f(s_i^{\mp 1}s_j^{\pm 1})^{f(b_1^{\mp 1}b_2^{\pm 1} \cdots b_{k_B}^{-1})}, \quad {\rm and} \\
(s'_j)^{-1} s'_i &=& f(s_j^{\pm 1}s_i^{\mp 1})^{f(b_1^{\pm 1}b_2^{\mp 1} \cdots b_{k_B -1}^{-1}) \cdot b_{k_B}}. 
\end{eqnarray*}
Now, an attacker might try to solve  the following two ${m\choose 2}$-simultaneous CP-instances:
\begin{eqnarray*}
 \{ (s'_i(s'_j)^{-1},f(s_i^{\mp 1}s_j^{\pm 1})) \mid 1 \le i \ne j\le m \}, \quad {\rm and} \\
 \{((s'_j)^{-1}s_i, f(s_j^{\pm 1}s_i^{\mp 1})) \mid 1 \le i\ne j \le m \}.
\end{eqnarray*}
Note that here the attacker has to solve both ${m\choose 2}$-simultaneous CP-instances.
If the center of $G$ is "small", the attacker might obtain the private keys 
$\tilde{b}_{lhs}=f(b_{k_B} \cdots b_2^{-\epsilon _B}b_1^{\epsilon _B})$ 
and $\tilde{b}_{rhs}=f(b_1^{\epsilon _B}b_2^{-\epsilon _B} \cdots b_{k_B-1}^{-1}) \cdot b_{k_B}$, where $\epsilon _B=1$ if $k_B$ odd and $-1$ otherwise.
Similarly, he might approach Alice's private keys by solving two ${n\choose 2}$-simultaneous CP-instances, 
thus possibly obtaining the corresponding keys $\tilde{a}_{lhs}=f(a_{k_A} \cdots a_2^{-\epsilon _A}a_1^{\epsilon _A})$ 
and $\tilde{a}_{rhs}=f(a_1^{\epsilon _A}a_2^{-\epsilon _A} \cdots a_{k_A-1}^{-1}) \cdot a_{k_A}$, and from
these (and $p_0$) $f(a_0)$. This suffices to recover the shared key
\begin{eqnarray*}
   K&=&f(a_{k_A} \cdots a_1^{\epsilon _A}b_{k_B}^{-\epsilon _A} \cdots b_1^{-\epsilon _A\epsilon _B} a_0^{\epsilon _A\epsilon _B}
          b_1^{-\epsilon _A\epsilon _B} \cdots  \cdot b_{k_B}^{-\epsilon _A} a_1^{\epsilon _A} \cdots a_{k_A-1}^{-1}) \cdot a_{k_A} \\
   &=&\tilde{a}_{lhs}\tilde{b}_{rhs} f(a_0)^{\epsilon _A\epsilon _B} f(\tilde{b}_{rhs}) \tilde{a}_{rhs}.
\end{eqnarray*}
Therefore, it is recommended to choose the generators $s_i$ (and $t_j$) of $S_A$ (and $S_B$) such that the following sets have large centralizers
\[ \{ f(s_is_j^{-1}) \mid 1 \le i \ne j\le m \}, \,\, \{ f(s_is_j^{-1}) \mid 1 \le i \ne j\le m \}, \]
\[ \{ f(t_it_j^{-1}) \mid 1 \le i \ne j\le n \}, \,\, \{ f(t_it_j^{-1}) \mid 1 \le i \ne j\le n \}. \]
This might  be achieved by choosing the $s_i$'s and $t_j$'s such that the generator sets of $S_A$ and $S_B$
have already large centralizers. 
\end{rem}

\subsubsection{Finite matrix groups as platforms}. 
(1) We propose matrix groups over the field of multivariate rational functions $F(t_1,\ldots , t_N)$ over $F=\mathbb{F}_q$ as possible
platform groups. The projecting endomorphism $f$ is an evaluation endomorphism, evaluating $M$ ($M\le N$) variables over the finite field $\mathbb{F}_q$.
More precisely, let $d\in \mathbb{N}$, $G=GL(d,F(t_1,\ldots , t_N))$, $I \in \{1, \ldots , N\}^M$ and $c \in \mathbb{F}_q^M$. 
Then $f=f_{I,c} \in End(G)$ is given by $t_{I_i} \mapsto c_i$ for all $i=1,\ldots , M$. Therefore, $(G, *_f)$ is our platform LD-system with $*_f$ being
$f$-symmetric conjugacy. 
\par All generators $s_i$ and $t_j$ should be chosen such that their images under the evaluation homomorphism $f$ are invertible,
and they should have large centralizers.
\par The large centralizer condition might be satisfied, for example, using the following construction. For $d=N$, we consider images of pure braids  under the
Gassner representation $P_N \longrightarrow GL(d,F(t_1,\ldots , t_N))$ \cite{Ga61} (where we reduce the involved integers modulo $q$).
Images of (conjugates of) reducible (or "cabled") pure braids will certainly have "large" centralizers.
\par 
Unfortunately, since the coefficient ring $F(t_1,\ldots , t_N)$ is infinite, the numerator and denominator polynomials start to grow quickly.
Thus, $G=GL(d,F(t_1,\ldots , t_N))$ is only for small parameter values an efficient platform group.
Nevertheless, as a challenge, we suggest, for example, the following parameter values. 
Let $d=4$, $M=N=1$, and $q =37$. 
For simplicity, we assume that $q$ is prime. Furthermore, we set $m=n=6$, and the iteration depths are $k_A=k_B=5$.
We set the number $l=l_A=l_B$ of internal nodes of the involved planar rooted binary trees to 5.
Recall that these trees are needed for the generation of $a_0,b_1, \ldots , b_{k_B}$. \par
\medskip
(2) The simultaneous iterated $f$-symmetric conjugator search problem appears to be hard even in finite groups.
 Here we propose a more efficient instantiation of the iterated $f$-symmetric conjugacy KEP in the finite matrix group
$G=GL(d, R)$ with coefficient ring $R=\mathbb{F}_{p}[X]/(X^N-1)$ ($N=p-1$) where the $f$-symmetric conjugacy operation is given by the 
homomorphism $f\in End(G)$ that is induced by the evaluation homomorphism $R \longrightarrow \mathbb{F}_p$, defined by $X\mapsto r$ for some fixed $r\in \mathbb{F}_p^*$.
This map is well defined since $r^{p-1}-1=0$ for all $r \in \mathbb{F}_p^*$ according to Fermat's little theorem. Though the ring $R$ has the same cardinality as $F_{p^{p-1}}$, $R$ is not a field since the polynomial $X^{p-1}-1= \prod _{r \in \mathbb{F}_p^*} (X-r)$ is not irreducible. For general $N$, $R$ is also called the \emph{ring of $N$-truncated polynomials}, and it is the platform ring of NTRUEncrypt \cite{HPS98}.
\par As a challenge, we suggest, for example, the following parameter values. 
Set $d=4$, $p =17$, $m=n=8$, and the iteration depths are $k_A=k_B=10$.
We set the number $l=l_A=l_B$ of internal nodes of the involved planar rooted binary trees also to 10.
\par More generally, we could have chosen $R=\mathbb{F}_q[X]/(g)$ as our coefficient ring, where $g$ is a reducible polynomial (of degree $N$) over $\mathbb{F}_q$ 
and $q$ is some prime power. Then $f\in End(G)$ is induced by some evaluation homomorphism on $R$ which evaluates $X$ on a root of $g$.
\medskip

{\bf Acknowledgements.} The first author acknowledges financial support by the Minerva Foundation of Germany. 
\par 
The second author acknowledges financial support by The Oswald Veblen Fund. 
We thank Boaz Tsaban for encouragement and fruitful discussions.

\end{document}